%% file: main.tex
\documentclass[letterpaper,twocolumn,10pt]{article}
\usepackage{usenix-2020-09}

\usepackage{tikz}
\usepackage{amsmath}

\usepackage{xspace}
\usepackage[]{collab}
\usepackage{subfig}
\usepackage{booktabs}
\usepackage{enumitem}
\usepackage[noend]{algpseudocode}
\usepackage{algorithm}
\usepackage{algorithmicx}
\usepackage{amsthm}
\usepackage{amsfonts} 
\newtheorem{definition}{Definition}
\newtheorem{lemma}{Lemma}
\newtheorem{theorem}{Theorem}
\usepackage[font=small]{caption}

\algnewcommand{\LineComment}[1]{\State // #1}

\usepackage[affil-it]{authblk}
\makeatletter
\renewcommand\AB@affilsepx{, \protect\Affilfont}
\makeatother
\usepackage{titlesec}
\titlespacing\section{0pt}{6pt plus 4pt minus 2pt}{4pt plus 2pt minus 2pt}
\titlespacing\subsection{0pt}{6pt plus 3pt minus 2pt}{4pt plus 2pt minus 2pt}
\titlespacing\subsubsection{0pt}{6pt plus 3pt minus 2pt}{4pt plus 2pt minus 2pt}

\newcommand{\kv}{FUSEE\xspace} 
\newcommand{\rep}{{{SNAPSHOT}}\xspace}

\begin{document}

\date{}

\title{\Large \bf \kv: A Fully Memory-Disaggregated Key-Value Store \\ (Extended Version)}

\author[1\thanks{Work mainly done during the internship at Huawei Cloud.}]{Jiacheng Shen}
\author[2]{Pengfei Zuo}
\author[3]{Xuchuan Luo}
\author[1]{Tianyi Yang}
\author[4]{\\ Yuxin Su}
\author[3]{Yangfan Zhou}
\author[1]{Michael R. Lyu}
\affil[1]{The Chinese University of Hong Kong}
\affil[2]{Huawei Cloud}
\affil[3]{Fudan University}
\affil[4]{Sun Yat-sen University}

\maketitle

\input{Sections/00-abstract}
\input{Sections/01-introduction}
\input{Sections/02-background}
\input{Sections/03-challenge}
\input{Sections/04-design}
\input{Sections/05-recovery}
\input{Sections/06-evaluation}
\input{Sections/07-literature}
\input{Sections/08-conclusion}
\input{Sections/09-acknowledgement.tex}

\bibliographystyle{plain}
\bibliography{main}

\appendix
\input{Sections/A-proof}

\end{document}

%% file: Sections/00-abstract.tex
\begin{abstract}
    Distributed in-memory key-value (KV) stores are embracing the disaggregated memory (DM) architecture for higher resource utilization.
    However, existing KV stores on DM employ a \emph{semi-disaggregated} design that stores KV pairs on DM but manages metadata with monolithic metadata servers, hence still suffering from low resource efficiency on metadata servers. 
    To address this issue, this paper proposes \kv, a \textit{\textbf{FU}lly memory-di\textbf{S}aggr\textbf{E}gated} KV Stor\textbf{\textit{E}} that brings disaggregation to metadata management.
    \kv replicates metadata, {\em i.e.}, the index and memory management information, on memory nodes, manages them directly on the client side, and handles complex failures under the DM architecture.
    To scalably replicate the index on clients, \kv proposes a client-centric replication protocol that allows clients to concurrently access and modify the replicated index.
    To efficiently manage disaggregated memory, \kv adopts a two-level memory management scheme that splits the memory management duty among clients and memory nodes.
    Finally, to handle the metadata corruption under client failures, \kv leverages an embedded operation log scheme to repair metadata with low log maintenance overhead.
    We evaluate \kv with both micro and YCSB hybrid benchmarks.
    The experimental results show that \kv outperforms the state-of-the-art KV stores on DM by up to $4.5$ times with less resource consumption.
\end{abstract}

%% file: Sections/01-introduction.tex
\section{Introduction}\label{sec:intro}
\noindent
Traditional in-memory key-value (KV) stores on monolithic servers have recently been ported to the disaggregated memory (DM) architecture for better resource efficiency~\cite{racehash,pdpm}.
Compared with monolithic servers, DM decouples the compute and memory resources into independent network-attached compute and memory pools~\cite{legoos,aifm,clio,concordia,MIND,zombieland,infiniswap,canfarmem}.
KV stores on DM can thus enjoy efficient resource pooling and have higher resource efficiency.

However, constructing KV stores on DM is challenging because the memory pool generally lacks the compute power to manage data and metadata.
Existing work~\cite{pdpm} proposes a \textit{semi-disaggregated} design that stores KV pairs in the disaggregated memory pool but retains metadata management on monolithic servers.
In such a design, the KV pair storage enjoys high resource utilization due to exploiting the DM architecture, but the metadata management does not. 
Many additional resources are exclusively assigned to the metadata servers in order to achieve high overall throughput~\cite{ChenLWL21,osdi06sage,sc14ren}.

To achieve full resource utilization, it is critical to bring disaggregation to the metadata management, \textit{i.e.}, building a \textit{fully memory-disaggregated} KV store.
The metadata, {\em i.e.}, the index and memory management information, should be stored in the memory pool and directly managed by clients rather than metadata servers. 
However, it is non-trivial to achieve a fully memory-disaggregated KV store due to the following challenges incurred from handling complex failures and the weak compute power in the memory pool.

\textit{1) Client-centric index replication.}
To tolerate memory node failures, clients need to replicate the index on memory nodes in the memory pool and guarantee the consistency of index replicas.
In existing replication approaches, {\em e.g.}, state machine replication~\cite{raft,epaxos,hermes,cr} and shared register protocols~\cite{robustEmulation,gryff,podc90hagit}, the replication protocols are executed by server-side CPUs.
These protocols cannot be executed on DM due to the weak compute power in the memory pool.
Meanwhile, if clients simply employ consensus protocols~\cite{raft,paxos,epaxos} or remote locks~\cite{pdpm}, the KV store suffers from poor scalability due to the explicit serialization of conflicting requests~\cite{scalingreplica,speedupconsensus,scalablerdmarpc,sherman}.

\textit{2) Remote memory allocation.}
Existing semi-disaggregated KV stores manage memory spaces with monolithic metadata servers.
However, in the fully memory-disaggregated setting, such a server-centric memory management scheme is infeasible.
Specifically, memory nodes cannot handle the compute-heavy fine-grained memory allocation for KV pairs due to their poor compute power~\cite{pdpm,clio}.
Meanwhile, clients cannot efficiently allocate memory spaces because multiple RTTs are required to modify the memory management information stored on memory nodes~\cite{MIND}.

\textit{3) Metadata corruption under client failures.}
In semi-disaggregated KV stores, client failures do not affect metadata because the CPUs of monolithic servers exclusively modify metadata.
However, clients directly access and modify metadata on memory nodes in the fully memory-disaggregated setting.
As a result, client failures can leave partially modified metadata accessible by others, compromising the correctness of the entire KV store.

To address these challenges, we propose \kv, a fully memory-disaggregated key-value store that has efficient index replication, memory allocation, and fault-tolerance on DM.
First, to maintain the strong consistency of the replicated index in a scalable manner, \kv proposes the \rep replication protocol.
The key to achieving scalability is to resolve write conflicts without involving the expensive request serialization~\cite{gryff}.
\rep adopts three simple yet effective conflict-resolution rules on clients to allow conflicts to be resolved collaboratively among clients instead of sequentially.
Second, to achieve efficient remote memory management, \kv employs a two-level memory management scheme that splits the server-centric memory management process into compute-light and compute-heavy tasks. 
The compute-light coarse-grained memory blocks are managed by the memory nodes with weak compute power, and the compute-heavy fine-grained objects are handled by clients.
Finally, to deal with the problem of metadata corruption, \kv adopts an embedded operation log scheme to resume clients' partially executed operations.
The embedded operation log reuses the memory allocation order and embeds log entries in KV pairs to reduce the log-maintenance overhead on DM.

We implement \kv from scratch and evaluate its performance using both micro and YCSB benchmarks~\cite{ycsb}.
Compared with Clover and pDPM-Direct~\cite{pdpm}, two state-of-the-art KV stores on DM, \kv achieves up to $4.5$ times higher overall throughput and exhibits lower operation latency with less resource consumption.
The code of \kv is available at \url{https://github.com/dmemsys/FUSEE}.

In summary, this paper makes the following contributions:
\begin{itemize}[noitemsep,topsep=0pt,parsep=0pt,partopsep=0pt]
    \item A fully memory-disaggregated KV store with disaggregated metadata and data that is resilient to failures on DM.
    \item A client-centric replication protocol that uses conflict resolution rules to enable clients to resolve conflicts collaboratively. The protocol is formally verified with TLA+~\cite{tla+} for safety and the absence of deadlocks under crash-stop failures.
    \item A two-level memory management scheme that leverages both memory nodes and clients to efficiently manage the remote memory space.
    \item An embedded operation log scheme to repair the corrupted metadata with low log maintenance overhead.
    \item The implementation and evaluation of \kv to demonstrate the efficiency and effectiveness of our design. 
\end{itemize}

%% file: Sections/02-background.tex
\section{Background and Motivation}\label{sec:background}
\subsection{The Disaggregated Memory Architecture}
\noindent
The disaggregated memory architecture is proposed to address the resource underutilization issue of traditional datacenters composed of monolithic servers~\cite{legoos,aifm,clio,MIND,concordia,zombieland}.
DM separates CPUs and memory of monolithic servers into two independent hardware resource pools containing compute nodes (CNs) and memory nodes (MNs)~\cite{legoos,racehash,sherman,pdpm}.
CNs have abundant CPU cores and a small amount of memory as local caches~\cite{sherman}.
MNs host various memory media, {\em e.g.}, DRAM and persistent memory, to accommodate different application requirements with weak compute power.
CPUs in CNs directly access memory in MNs with fast remote-access interconnect techniques, such as one-sided RDMA (remote direct memory access), Omni-path~\cite{omnipath}, CXL~\cite{cxl}, and Gen-Z~\cite{genz}.
Each MN provides \texttt{READ}, \texttt{WRITE}, and atomic operations, {\em i.e.}, compare-and-swap (\texttt{CAS}) and fetch-and-add (\texttt{FAA}), for CNs to access memory data.
Besides, MNs own limited compute power ({\em e.g.}, 1-2 CPU cores) to manage local memory and establish connections from CNs, providing CNs with the \texttt{ALLOC} and \texttt{FREE} memory management interfaces.
Without loss of generality, in this paper, we consider CNs accessing MNs using one-sided RDMA verbs.

\begin{figure}
    \centering
    \subfloat[Clover]{
        \includegraphics[width=0.55\columnwidth]{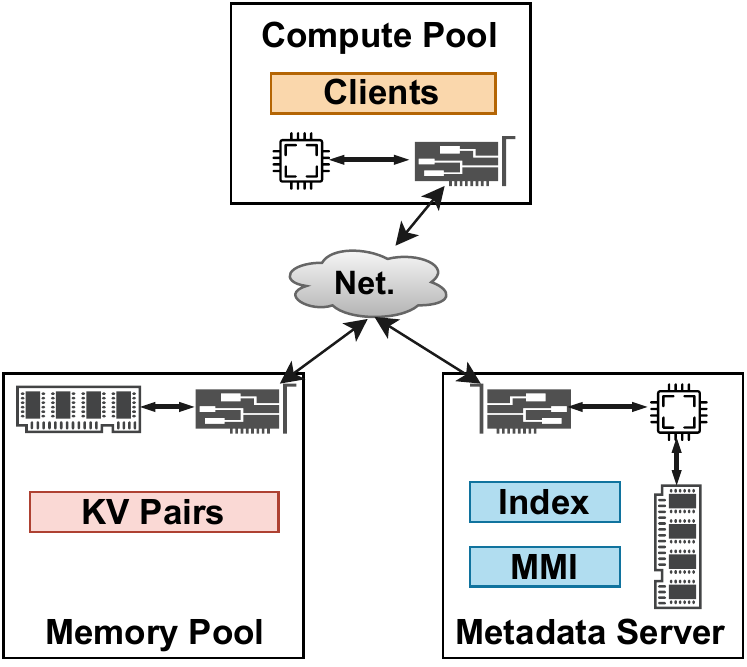}
        \label{fig:clover-arch}
    }
    \hspace{0.8cm}
    \subfloat[\kv]{
        \includegraphics[width=0.22\columnwidth]{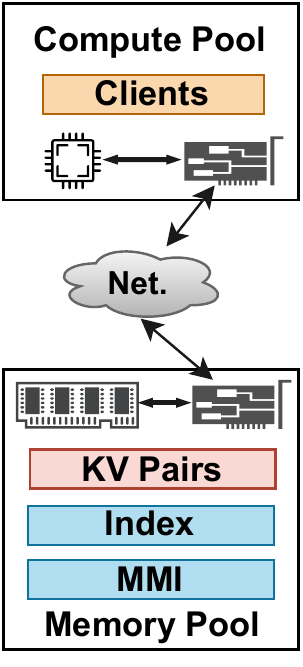}
        \label{fig:decent-arch}
    }
    \caption{Two architectures of memory-disaggregated KV stores. (a) The semi-disaggregated architecture (Clover~\cite{pdpm}). (b) The fully disaggregated architecture proposed in this paper.}
    \label{fig:memdisaggkv-arch}
    \vspace{-0.25in}
\end{figure}

\subsection{KV Stores on Disaggregated Memory}
\noindent
Clover~\cite{pdpm} is a state-of-the-art KV store built on DM.
It adopts a semi-disaggregated design that separates data and metadata to lower the ownership cost and prevent the compute power of data nodes from becoming the performance bottleneck.
As shown in Figure~\ref{fig:clover-arch}, Clover deploys clients on CNs and stores KV pairs on MNs.
It adopts additional monolithic metadata servers to manage the metadata, including \textit{memory management information (MMI)} and the \textit{hash index}.
For \texttt{SEARCH} requests, clients look up the addresses of the KV pairs from metadata servers and then fetch the data on MNs using \texttt{RDMA\_READ} operations.
For \texttt{INSERT} and \texttt{UPDATE} requests, clients allocate memory blocks from metadata servers with RPCs, write KV pairs to MNs with \texttt{RDMA\_WRITE} operations, and update the hash index on the metadata servers through RPCs.
To prevent clients' frequent requests from overwhelming the metadata servers, clients allocate a batch of memory blocks one at a time and cache the hash index locally.
As a result, Clover achieves higher throughput under read-intensive workloads with less resource consumption.

However, the semi-disaggregated design of Clover cannot fully exploit the resource efficiency of the DM architecture due to its monolithic-server-based metadata management.
On the one hand, monolithic metadata servers consume additional resources, including CPUs, memory, and RNICs.
On the other hand, many compute and memory resources have to be reserved and assigned to the metadata server of Clover to achieve good performance due to the CPU-intensive nature of metadata management~\cite{ChenLWL21,osdi06sage,sc14ren}.
To show the resource utilization issue of Clover, we evaluate its throughput with 2 MNs, 64 clients, and a metadata server with different numbers of CPU cores.
We control the number of CPU cores by assigning different percentages of CPU time with cgroup~\cite{cgroups}.
As shown in Figure~\ref{fig:rd-wr-tpt}, Clover has a low overall throughput with a small number of CPU cores assigned to its metadata server.
At least six additional cores have to be assigned until the metadata server is no longer the performance bottleneck.

To attack the problem, \kv adopts a \textit{fully memory-disaggregated} design that enables clients to directly access and modify the hash index and manage memory spaces on MNs, as shown in Figure~\ref{fig:decent-arch}.
Compared with the semi-disaggregated design, resource efficiency can be improved because client-side metadata management eliminates the additional metadata servers.
The overall throughput can also be improved because the computation bottleneck of metadata management no longer exists.

%% file: Sections/03-challenge.tex
\section{Challenges}
\noindent
This section introduces the three challenges of constructing a fully memory-disaggregated KV store, {\em i.e.}, index replication, remote memory allocation, and metadata corruption.

\subsection{Client-Centric Index Replication}
\noindent
The index must be replicated to tolerate MN failures.
Strong consistency, {\em i.e.}, linearizability~\cite{lin}, is the most commonly adopted correctness standard for data replication because it reduces the complexity of implementing upper-level applications~\cite{gryff,hotos15phillipe,atc17yu}.
Linearizability requires that operations on an object appear to be executed in some total order that respects the operations' real-time order~\cite{lin}.
The key challenge of achieving a linearizable replicated hash index under the fully memory-disaggregated setting comes from the client-centric computation nature of DM.

\begin{figure}[t]
    \hspace{-1mm}
    \begin{minipage}[t]{.49\columnwidth}
        \centering
        \includegraphics[width=\columnwidth]{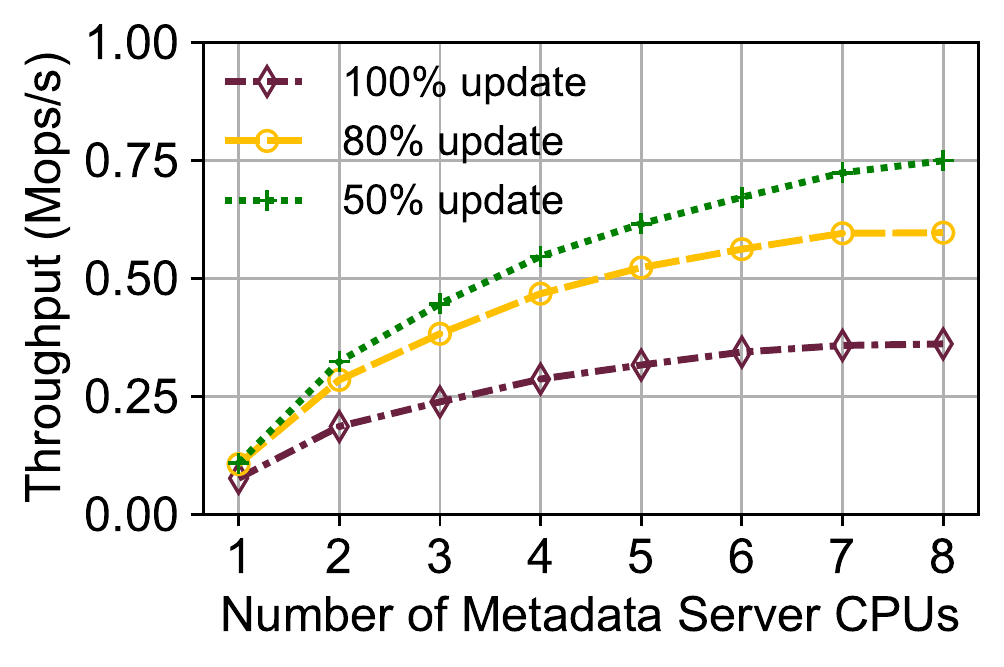}
        \caption{The throughput of Clover with an increasing number of metadata server CPUs.}
        \label{fig:rd-wr-tpt}
    \end{minipage}%
    \hspace{3mm}
    \begin{minipage}[t]{.46\columnwidth}
        \centering
        \includegraphics[width=\columnwidth]{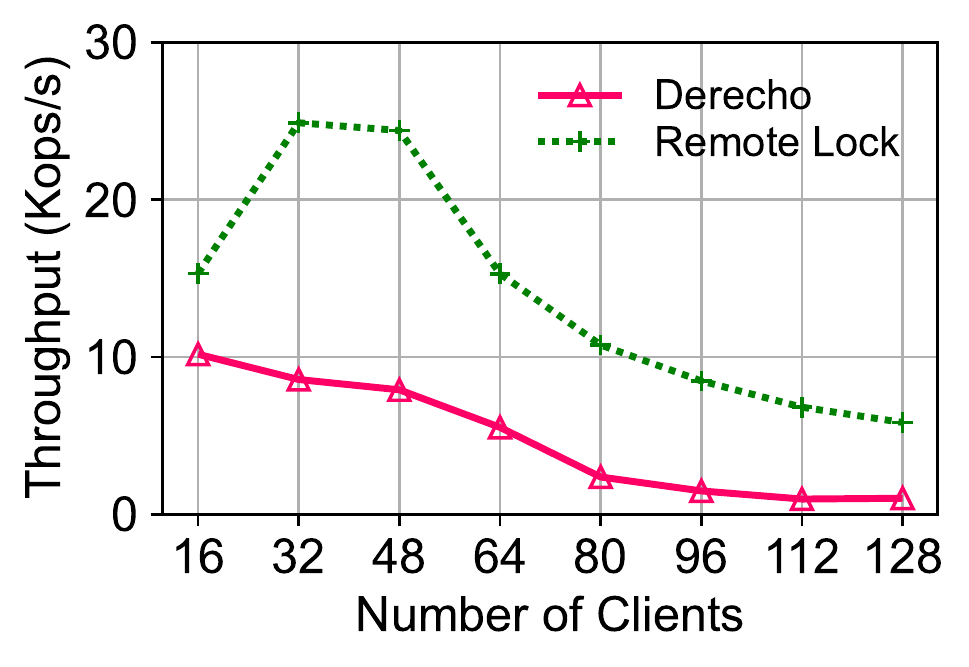}
        \caption{The throughput of Derecho~\cite{tocs19Sagar} and lock-based approaches.}
        \label{fig:scale-issue}
    \end{minipage}
    \vspace{-0.2in}
\end{figure}

First, existing replication methods are not applicable in the fully memory-disaggregated setting due to their server-centric nature.
State machine replication (SMR)~\cite{raft,epaxos,hermes,cr,craq,mencius,viewstamp} and shared register protocols~\cite{gryff,robustEmulation} are two major replication approaches that achieve linearizability.
However, both approaches are designed with a server-centric assumption that a data replica is exclusively accessed and modified by the CPU that manages the data.
First, the SMR approaches consider the CPU and the data replica as a state machine and achieve strong consistency by forcing the state machines to execute deterministic KV operations in the same global order~\cite{raft,viewstamp}.
Server CPUs are extensively used to reach a consensus on a global operation order and apply state transitions to data replicas.
Second, shared register protocols view the CPU and the data replica as a shared register with \texttt{READ} and \texttt{WRITE} interfaces.
Linearizability is achieved with a last-writer-wins conflict resolution scheme~\cite{robustEmulation} that forces a majority of shared registers to always hold data with the newest timestamps.
Shared register protocols also heavily rely on server-side CPUs to compare timestamps and apply data updates.
The challenge with the server-centric approaches is that in the fully memory-disaggregated scenario, there is no such management CPU because all clients directly access and modify the hash index with one-sided RDMA verbs.

Second, naively adopting consensus protocols or remote locks among clients results in poor throughput due to the expensive request serialization.
To show the performance issues of consensus protocols and remote locks, we store and replicate a shared object on two MNs and vary the number of concurrent clients.
We use a state-of-the-art consensus protocol Derecho~\cite{tocs19Sagar} and an RDMA \texttt{CAS}-based spin lock to ensure the strong consistency of the replicated object.
As shown in Figure~\ref{fig:scale-issue}, both Derecho and lock-based approaches exhibit poor overall throughput and cannot scale with the growing number of concurrent clients.

\subsection{Remote Memory Allocation}\label{sec:challenge-rma}
\noindent
The key challenge of managing DM is where to execute the memory-management computation.
There are two possible DM management approaches~\cite{MIND}, {\em i.e.}, compute-centric ones and memory-centric ones.
The compute-centric approaches store the memory management metadata on MNs and allow clients to allocate memory spaces by directly modifying the on-MN metadata.
Since the memory management metadata are shared by all clients, clients' accesses have to be synchronized.
As a result, compute-centric approaches suffer from the high memory allocation latency incurred by the expensive and complex remote synchronization process on DM~\cite{MIND}.
The memory-centric approaches maintain all memory management metadata on MNs with their weak compute power.
Such approaches are also infeasible because the poor memory-side compute power can be overwhelmed by the frequent fine-grained KV allocation requests from clients.
Although there are several approaches that conduct memory management on DM, they all target page-level memory allocation and rely on special hardware, {\em i.e.}, programmable switches~\cite{MIND} and SmartNICs~\cite{clio}, which are orthogonal to our problem.

\subsection{Metadata Corruption}
\noindent
In fully memory-disaggregated KV stores, crashed clients can leave partially modified metadata accessible by other healthy clients.
Since the metadata contains important system state, metadata corruption compromises the correctness of the entire KV store.    
First, crashed clients may leave the index in a partially modified state.
Other healthy clients may not be able to access data or even access wrong data with the corrupted index.
Second, crashed clients may allocate memory spaces but not use them, causing severe memory leakage.
Hence, in order to ensure the correctness of the KV store, the corrupted metadata has to be repaired under client failures.

%% file: Sections/04-design.tex
\section{The \kv Design}\label{sec:design}

\subsection{Overview}\label{sec:overview}
\noindent
As shown in Figure~\ref{fig:ddckv-arch}, \kv consists of clients, MNs, and a master.
Clients provide \texttt{SEARCH}, \texttt{INSERT}, \texttt{DELETE}, and \texttt{UPDATE} interfaces for applications to access KV pairs.
MNs store the replicated memory management information (MMI), hash index, and KV pairs.
The master is a cluster management process responsible only for initializing clients and MNs and recovering data under client and MN failures.


\begin{figure}[t]
    \centering
    \includegraphics[width=0.75\columnwidth]{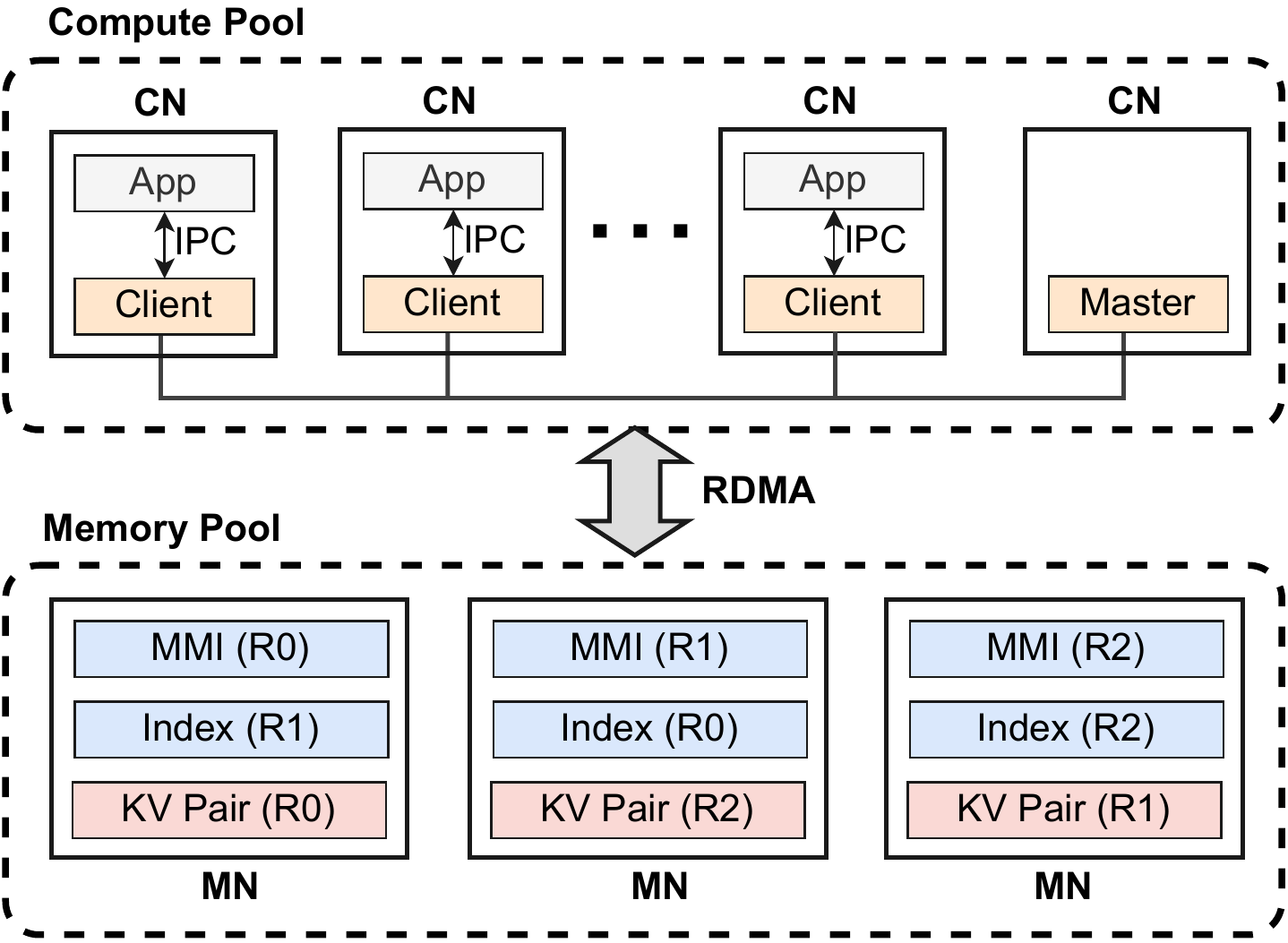}
    \caption{The \kv overview \textit{(MMI, Index, and KV pairs have multiple replicas, i.e., $R_0$, $R_1$, and $R_2$. $R_0$ is the primary replica.)}.}
    \label{fig:ddckv-arch}
    \vspace{-0.1in}
\end{figure}

\begin{figure}[t]
    \centering
    \includegraphics[width=0.75\columnwidth]{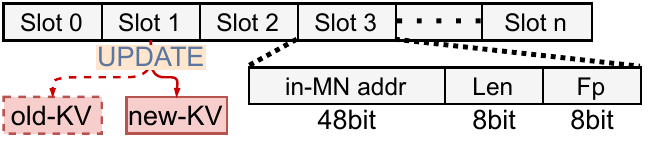}
    \caption{The structure of an index replica.}
    \label{fig:index-structure}
    \vspace{-0.25in}
\end{figure}

\kv replicates both the hash index and KV pairs to tolerate MN failures.
We adopt RACE hashing (Section~\ref{sec:racehash}) to index KV pairs and propose the \rep replication protocol to enforce the strong consistency of the replicated hash index (Section~\ref{sec:rep}).
A two-level memory management scheme is adopted to efficiently allocate and replicate variable-sized KV pairs (Section~\ref{sec:2lmm}).
\kv uses logs to handle the corrupted metadata under client failures and adopts an embedded operation log scheme to reduce the log maintenance overhead (Section~\ref{sec:embedlog}).
Other optimizations are introduced in Section~\ref{sec:opt} to further improve the system performance.

\subsection{RACE Hashing}\label{sec:racehash}
\noindent
RACE hashing is a one-sided RDMA-friendly hash index.
As shown in Figure~\ref{fig:index-structure}, it contains multiple 8-byte slots, with each storing a pointer referring to the address of a KV pair, an 8-bit fingerprint (Fp), {\em i.e.}, a part of the key's hash value, and the length of the KV pair (Len)~\cite{racehash}.
For \texttt{SEARCH} requests, a client reads the slots of the hash index according to the hash value of the target key and then reads the KV pair on MNs according to the pointer in the slot.
For \texttt{UPDATE}, \texttt{INSERT}, and \texttt{DELETE} requests, RACE hashing adopts an \emph{out-of-place modification} scheme.
It first writes a KV pair to MNs and then modifies the corresponding slot in the hash index to the address of the KV pair atomically with an \texttt{RDMA\_CAS}.
Nevertheless, the RACE hashing only supports a single replica.

\begin{figure}
    \centering
    \includegraphics[width=\columnwidth]{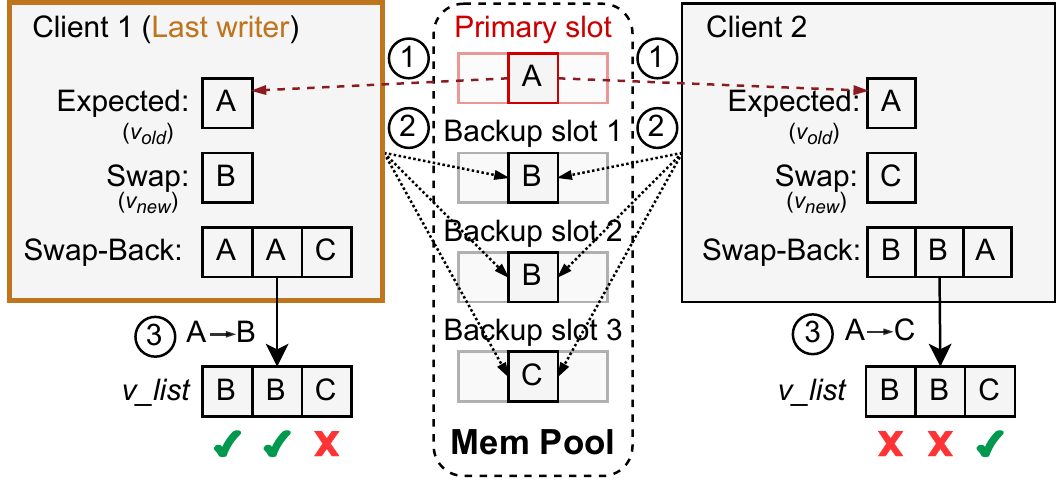}
    \caption{The \rep replication protocol.}
    \label{fig:replication}
    \vspace{-0.25in}
\end{figure}

\subsection{The \rep Replication Protocol}\label{sec:rep}
\noindent
In \kv, multiple clients concurrently read or write the same slot in the replicated hash index to execute \texttt{SEARCH} or \texttt{UPDATE} requests, as shown in Figure~\ref{fig:replication}. 
To efficiently maintain the strong consistency of slot replicas in the replicated hash index, \kv proposes the \rep replication protocol, a client-centric replication protocol that achieves linearizability without the expensive request serialization.

There are two main challenges to efficiently achieving linearizability under the fully memory-disaggregated setting.
First, how to protect readers from reading incomplete states during read-write conflicts.
Second, how to resolve write-write conflicts without expensively serializing all conflicting requests.
To address the first challenge, \rep splits the replicated hash index into a single primary replica and multiple backup replicas and uses backup replicas to resolve write conflicts.
Hence, incomplete states during write conflicts only appear on backup replicas and the primary replica always contains the correct and complete value. 
Readers can simply read the contents in the primary replica without perceiving the incomplete states.
To address the second challenge, \rep adopts a last-writer-wins conflict resolution scheme similar to shared register protocols.
\rep leverages the \emph{out-of-place modification} characteristic of RACE hashing that conflicting writers always write different values into the same slot because the values are pointers referring to KV pairs at different locations.
Three conflict-resolution rules are thus defined based on the values written by conflicting writers in backup replicas, which enable clients collaboratively to decide on a single last writer under write conflicts.

\input{Algorithms/SNAPSHOT-no-failure}

Algorithm~\ref{alg:rep} shows the \texttt{READ} and \texttt{WRITE} processes of the \rep replication protocol.
Here we focus on the execution of \rep when no failure occurs and leave the discussion of failure handling in Section~\ref{sec:execution}.
We call the slots in the primary and backup hash indexes primary slots and backup slots, respectively.

For \texttt{READ} operations, clients directly read the values in the primary slots using \texttt{RDMA\_READ}.
For \texttt{WRITE} operations, \rep first resolves write conflicts by letting conflicting writers collaboratively decide on a last writer with three conflict resolution rules and then let the decided last writer modify the primary slot.
Figure~\ref{fig:replication} shows the process that two clients simultaneously \texttt{WRITE} the same slot.
The corresponding algorithms are shown in Algorithms~\ref{alg:rep} and~\ref{alg:rep-eval}.
Clients first read the value in the primary slot as $v_{old}$ (\textcircled{\footnotesize 1}).
Then each client modifies all backup slots by broadcasting \texttt{RDMA\_CAS} operations (\textcircled{\footnotesize 2}) with $v_{old}$ as the expected value and $v_{new}$ as the swap value.
On receiving an \texttt{RDMA\_CAS}, the RNICs on MNs atomically modify the value in the target slot only if $v_{old}$ matches the current value in the slot.
Since all writers initiate \texttt{RDMA\_CAS} operations with the same $v_{old}$ and different $v_{new}$s and all backup slots initially hold $v_{old}$, the atomicity of \texttt{RDMA\_CAS} ensures that each backup slot can only be modified once by a single writer.
As a result, the values in all backup slots will be fixed after each of them has received one \texttt{RDMA\_CAS} from one writer~\footnote{That the process that all conflicting clients broadcast \texttt{RDMA\_CAS}es to modify backup slots is just like taking a snapshot, which is why the replication protocol is named \rep.}.
Meanwhile, since an \texttt{RDMA\_CAS} returns the value in the slot before it is modified, all clients can perceive the new values in the backup slots (\textcircled{\footnotesize 3}) through the return values of the broadcast of \texttt{RDMA\_CAS} operations. 
The return values are denoted as $v\_list$ in Algorithm~\ref{alg:rep}.

With $v\_list$, \rep defines the following three rules to let conflicting clients collaboratively decide on a last writer:
\begin{itemize}[noitemsep,topsep=0pt,parsep=0pt,partopsep=0pt]
    \item[] \textbf{Rule 1}: A client that has successfully modified all the backup slots is the last writer.
    \item[] \textbf{Rule 2}: A client that has successfully modified a majority of backup slots is the last writer.
    \item[] \textbf{Rule 3}: If no last writer can be decided with the former two rules, the client that has written the minimal target value ($v_{new}$) is considered as the last writer.
\end{itemize}
The three rules are evaluated sequentially as shown in Algorithm~\ref{alg:rep-eval}.
\textbf{Rule 1} provides a fast path when there are no conflicting modifications.
\textbf{Rule 2} preserves the most successful CAS operations to minimize the overhead of executing atomic operations on RNICs when conflicts are rare~\cite{designGuideline}.
Finally, \textbf{Rule 3} ensures that the protocol can always decide on the last writer.
To ensure the uniqueness of the last write, a client issues another \texttt{RDMA\_READ} to check if the primary slot has been modified (Line 12, Algorithm~\ref{alg:rep-eval}) before evaluating \textbf{Rule 3}.
If the primary slot has not been modified, then the \texttt{RDMA\_CAS\_backups} (Line~7, Algorithm~\ref{alg:rep}) of the client must happen before the last writer modifies the primary slot.
Hence, it is safe to evaluate \textbf{Rule 3} because the $v\_list$ must contain the value of the last writer if it has already been decided.
Otherwise, \textbf{Rule 3} will not be evaluated because the modification of the primary slot means the decision of a last writer.
Relying on the three rules, a unique last writer can be decided without any further network communications.
For example, in Figure~\ref{fig:replication}, Client 1 is the last writer according to \textbf{Rule 2}.
Client 1 then modifies the backup slots that do not yet contain its proposed value using \texttt{RDMA\_CAS}es and then modifies the primary slot.
Other conflicting clients iteratively \texttt{READ} the value in the primary slot and return success after finding the change in the primary slot.
The primary slot may remain unmodified only under the situation when the last writer crashed, which will be discussed in Section~\ref{sec:execution}.

\textbf{Correctness.}
The \rep replication protocol guarantees linearizability of the replicated hash indexes with last-writer-wins conflict resolution like shared register protocols~\cite{gryff,robustEmulation}.
We briefly demonstrate the correctness of \rep using the notion of the linearizable point of KV operations.
A formal proof is shown in Appendix~\ref{sec:proof}.
A linearizable point is a point when an operation atomically takes effect in its invocation and completion~\cite{lin}.
For \texttt{READ}, the linearizable point happens when it gets the value in the primary slot.
For \texttt{WRITE} operations, the linearizable point of the last writer happens when it modifies the primary slot.
Linearizable points of other conflicting writers appear instantly before the last writer modifies the primary slot.
Conflicts between readers and the last writer are resolved by RNICs because the last writer atomically modifies the primary slot using \texttt{RDMA\_CAS} operations and readers access the primary slot using \texttt{RDMA\_READ} operations.

\textbf{Performance.}
\rep guarantees a bounded worst-case latency when clients \texttt{WRITE} the hash index.
Under the situation when \textbf{Rule~1} is triggered, 3 RTTs are required to finish a \texttt{WRITE} operation.
Under situations when \textbf{Rule~2} or \textbf{Rule~3} is triggered, 4 or 5 RTTs are required, respectively.

\begin{figure}[t]
    \centering
    \includegraphics[width=0.77\columnwidth]{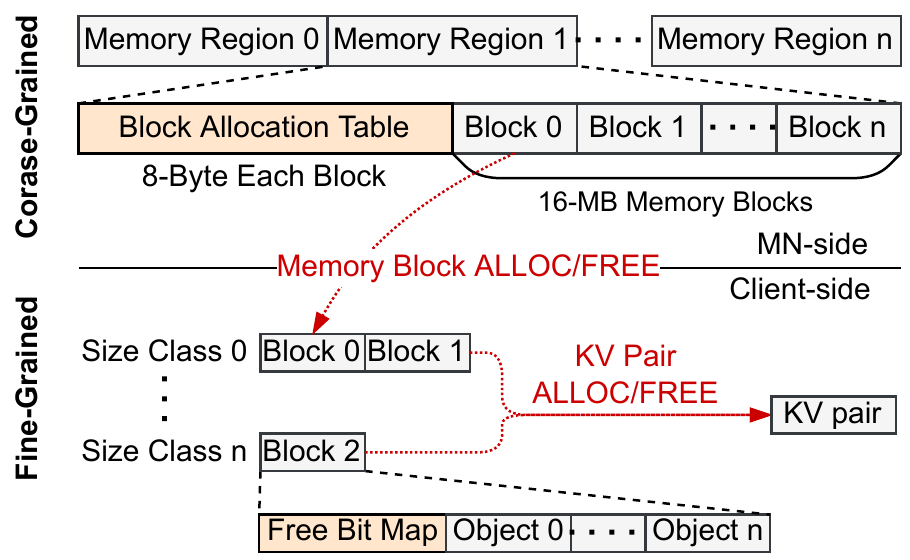}
    \caption{The two-level memory management scheme.}
    \label{fig:2lmm}
    \vspace{-0.25in}
\end{figure}

\subsection{Two-Level Memory Management}\label{sec:2lmm}
\noindent
Memory management is responsible for allocating, replicating, and freeing memory spaces for KV pairs on MNs.
As discussed in Section~\ref{sec:challenge-rma}, the key challenge of DM management is that conducting the management tasks solely on clients or on MNs cannot satisfy the performance requirement of frequent memory allocation for KV requests.
\kv addresses this issue via a two-level memory management scheme.
The key idea is to split the server-centric memory management tasks into compute-light coarse-grained management and compute-heavy fine-grained management run on MNs and clients.

\kv first replicates and partitions the 48-bit memory space on multiple MNs.
Similar to FaRM~\cite{farm}, \kv shards the memory space into 2GB memory regions and maps each region to $r$ MNs with consistent hashing~\cite{consistenthashing}, where $r$ is the replication factor.
Specifically, consistent hashing maps a region to a position in a hash ring. 
The replicas are then stored at the $r$ MNs successively following the position and the primary region is placed on the first of the $r$ MN.

Figure~\ref{fig:2lmm} shows the two-level memory allocation of \kv.
Allocating a memory space for a KV pair happens before writing the KV pair, as introduced in Section~\ref{sec:overview}.
The first level is the coarse-grained MN-side memory block allocation with low computation requirements.
Each MN partitions its local memory regions into coarse-grained memory blocks, {\em e.g.}, 16 MB, and maintains a block allocation table ahead of each region.
For each memory block, the block allocation table records a client ID (CID) that allocates it.
Clients allocate memory blocks by sending \texttt{ALLOC} requests to MNs.
On receiving an \texttt{ALLOC} request, an MN allocates a memory block from one of its primary memory regions, records the client ID in the block allocation tables of both primary and backup regions, and replies with the address of the memory block to the client.
The coarse-grained memory allocation information is thus replicated on $r$ MNs and can survive MN failures.
The second level is the fine-grained client-side object allocation that allocates small objects to hold KV pairs.
Specifically, clients manage the blocks allocated from MNs exclusively with slab allocators~\cite{slab}.
The client-side slab allocators split memory blocks into objects of distinct size classes.
A KV pair is then allocated from the smallest size class that fits it.

The allocated objects can be freed by any client.
To efficiently reclaim freed memory objects on client sides, \kv stores a free bit map ahead of each memory block, as shown in Figure~\ref{fig:2lmm}, where each bit indicates the allocation state of one object in the memory block.
The free bit map is initialized to be all zeros when a block is allocated.
To free an object, a client sets the corresponding bit to `1' in the free bit map with an \texttt{RDMA\_FAA} operation.
By reading the free bit map, clients can easily know the freed objects in their memory blocks and reclaim them locally.
\kv frees and reclaims memory objects periodically using background threads in a batched manner to avoid the additional RDMA operations on the critical paths of KV accesses.
The correctness of concurrently accessing KV pairs and reclaiming memory spaces is guaranteed by RACE hashing~\cite{racehash}, where clients check the key and the CRC of the KV pair on data accesses.

\begin{figure}
    \centering
    \subfloat[The embedded log entry.] {
        \includegraphics[width=0.75\columnwidth]{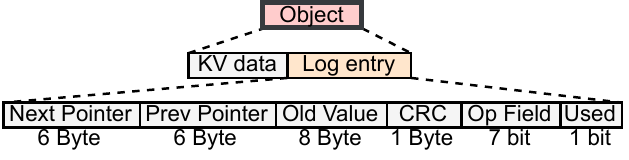}
        \label{fig:embedlog-entry}
    }\\
    \vspace{-0.15in}
    \subfloat[The organization of the embedded operation log.]{
        \includegraphics[width=0.73\columnwidth]{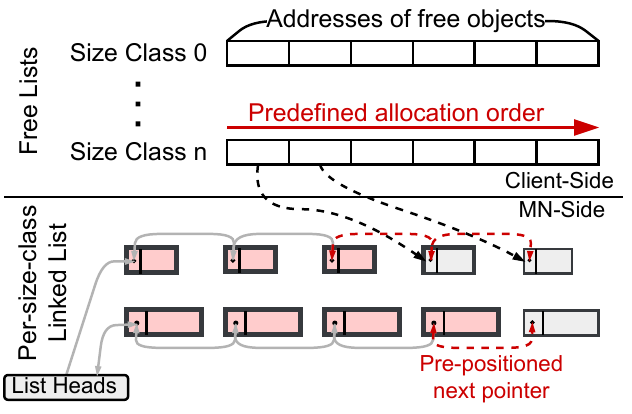}
        \label{fig:embedlog-org}
    }
    \caption{The embedded operation log.}\label{fig:embedlog}
    \vspace{-0.25in}
\end{figure}

\subsection{Embedded Operation Log}\label{sec:embedlog}
\noindent
Operation logs are generally adopted to repair the partially modified hash index incurred by crashed clients.
Conventional operation logs record a log entry for each KV request that modifies the hash index.
The log entries are generally written in an append-only manner so that the order of log entries reflects the execution order of KV requests.
The recovery process can thus find the crashed request and fix the corrupted metadata by scanning the ordered log entries.
However, constructing operation logs incurs high log maintenance overhead on DM because writing log entries adds remote memory accesses on the critical paths of KV requests.

To reduce the log maintenance overhead on DM, \kv adopts an \emph{embedded operation log} scheme that embeds log entries into KV pairs.
The embedded log entry is written together with its corresponding KV pair with one \texttt{RDMA\_WRITE} operation.
The additional RTTs required for persisting log entries are thus eliminated.
However, by embedding log entries in KV pairs, the execution order of KV requests cannot be maintained because the log entries are no longer continuous.
To address this problem, the embedded operation log scheme maintains per-size-class linked lists to organize the log entries of a client in the execution order of KV requests.
As shown in Figure~\ref{fig:embedlog-org}, a per-size-class linked list is a doubly linked list that links all allocated objects of the size class in the order of their allocations.
The object allocation order reflects the execution order of KV requests because all KV requests that modify the hash index, {\em e.g.}, \texttt{INSERT} and \texttt{UPDATE}, need to allocate objects for new KV pairs.
For \texttt{DELETE}, \kv allocates a temporary object recording the log entry and the target key and reclaims the object on finishing the \texttt{DELETE} request.
\kv stores the list heads on MNs during the initialization of clients, which will be accessed during the recovery process of clients (Section~\ref{sec:execution}).

\begin{figure*}
    \centering
    \includegraphics[width=1.7\columnwidth]{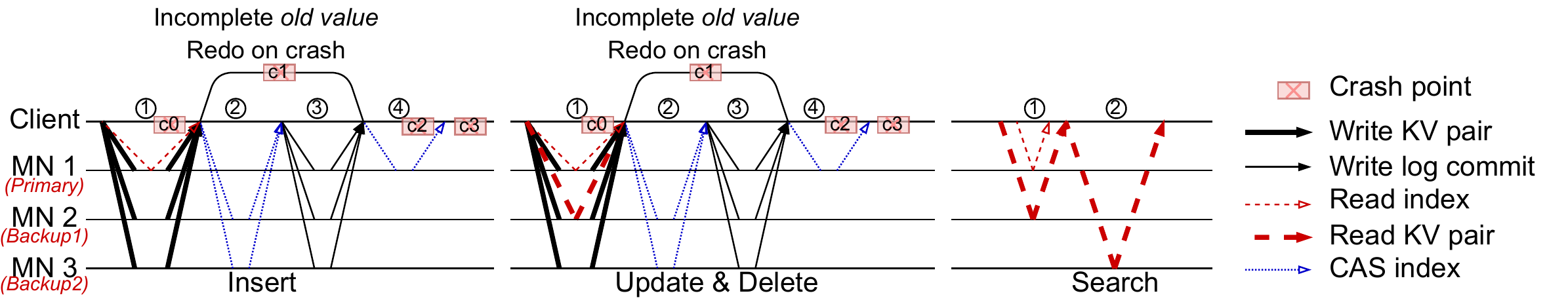}
    \caption{The workflows of different KV requests. {\it \small \texttt{INSERT}: \textcircled{\scriptsize 1} write the KV pair to all replicas and read the primary index slot. \textcircled{\scriptsize 2} CAS all backup slots. \textcircled{\scriptsize 3} write the old value to the log header. \textcircled{\scriptsize 4} CAS the primary slot. \texttt{UPDATE} \& \texttt{DELETE}: \textcircled{\scriptsize 1} write the KV pair, read the primary slot, and read the KV pair according to the index cache. \textcircled{\scriptsize 2} CAS backup slots. \textcircled{\scriptsize 3} write the old value to the log header. \textcircled{\scriptsize 4} CAS the primary slot. \texttt{SEARCH}: \textcircled{\scriptsize 1} read the primary slot and the KV pair according to the index cache. \textcircled{\scriptsize 2} read the KV pair on cache misses.}}
    \label{fig:kv-ops}
    \vspace{-0.25in}
\end{figure*}

An embedded log entry is a 22-byte data structure stored behind KV pairs, as shown in Figure~\ref{fig:embedlog-entry}.
It contains a 6-byte \textit{next pointer}, a 6-byte \textit{prev pointer}, an 8-byte \textit{old value}, a 1-byte \textit{CRC}, a 7-bit \textit{opcode}, and a \textit{used} bit.
The \textit{next pointer} points to the next object of the size class that will be allocated and the \textit{prev pointer} points to the object allocated before the current one.
The \textit{old value} records the old value of the primary slot for recovery proposes, which will be discussed in Section~\ref{sec:execution}.
The 1-byte \textit{CRC} is adopted to check the integrity of the \textit{old value} under client failures.
The \textit{operation field} records the operation type, {\em i.e.}, \texttt{INSERT}, \texttt{UPDATE}, or \texttt{DELETE}, so that the crashed operation can be properly retried during recovery.
The \textit{used bit} indicates if an object is in-use or free.
Storing the \textit{used bit} at the end of the entire object can be used to check the integrity of an entire object. 
This is because the order-preserving nature of \texttt{RDMA\_WRITE} operations ensures that the used bit is written only after all other contents in the object have been successfully written.

\kv efficiently organizes per-size-class linked lists by co-designing the linked list maintenance process with the memory allocation process.
As shown in Figure~\ref{fig:embedlog-org}, for each size class, a client organizes the addresses of remote free objects locally as a free list.
Since an object is always allocated from the head of a local free list, the allocation order of each size class is pre-determined.
Based on the pre-determined order, for each allocation, a client pre-positions the \textit{next pointer} to point to the free object in the head of the local free list and the \textit{prev pointer} to point to the last allocated object of the size class.
Both the \textit{next pointer} and the \textit{prev pointer} are thus known before each allocation and the entire log entry can be written to MNs with the KV pair in a single \texttt{RDMA\_WRITE}.

Combined with the \rep replication protocol, the execution process is shown as follows.
First, for each writer, a log entry with an empty \textit{old value} and \textit{CRC} is written with the KV pair in a single \texttt{RDMA\_WRITE}.
Then, for the last writer of the \rep replication protocol, the \textit{old value} is modified to store the old value of the primary slot before the primary slot is modified.
For other non-last writers, the used bits in their corresponding KV log entries are reset to `0' after finding the modification of the primary slot.

\subsection{Optimizations}\label{sec:opt}
\noindent
\textbf{Adaptive index cache.}
Index caching is widely adopted on RDMA-based KV stores to reduce request RTTs~\cite{drtm,drtmh,xstore,pdpm}.
For a key, the index cache caches the remote addresses of the replicated index slots and the addresses of the KV pairs locally.
With the cached KV pair addresses, \texttt{UPDATE}, \texttt{DELETE}, and \texttt{SEARCH} requests can read KV pairs in parallel with searching the hash index, reducing an RTT on cache hits.
To guarantee cache coherence, an invalidation bit is stored together with each KV pair, which is used by clients to check whether the KV pair is valid or invalid. However, by accessing the index cache, invalid KV pairs ({\em e.g.}, outdated) can be fetched into clients, causing read amplification.

To attack the read amplification issue, \kv adaptive bypasses the index cache by distinguishing read-intensive and write-intensive keys.
For each cached key, \kv maintains an access counter and an invalid counter which increases by 1 each time the key is accessed or found to be invalid. 
A client calculates an \textit{invalid ratio} $I=\frac{\text{invalid counter}}{\text{access counter}}$ for each cached key.
The index cache is bypassed when accessing a key with $I > threshold$ because the key is write-intensive and the cached key address points to an invalid KV pair with high probability.
The \textit{invalid ratio} can adapt to workload changes, {\em i.e.}, a write-intensive key becomes read-intensive, because the access counter of the key keeps increasing while the invalid counter stops.
Besides, the adaptive scheme does not affect the \texttt{SEARCH} latency for most cases since only write-intensive keys bypass the cache.

\noindent
\textbf{RDMA-related optimizations.}
KV requests require multiple remote memory accesses.
\kv adopts doorbell batching and selective signaling~\cite{designGuideline} to reduce RDMA overhead.
Figure~\ref{fig:kv-ops} shows the procedures for executing different KV requests.
Each request consists of multiple phases with multiple network operations.
For each phase, \kv adopts doorbell batching~\cite{designGuideline} to reduce the overhead of transmitting network operations from user space to RNICs and selective signaling to reduce the overhead of polling RDMA completion queues. 
Consequently, each phase only incurs 1 network RTT.
For \texttt{INSERT}, \texttt{DELETE}, and \texttt{UPDATE} requests, four RTTs are required in general cases.
For \texttt{SEARCH} requests, at most two RTTs are required and only one RTT is required in the best case due to the index cache.

%% file: Algorithms/SNAPSHOT-no-failure.tex
\begin{figure}[t]
\begin{algorithm}[H]
\begin{algorithmic}[1]
\caption{The \rep replication protocol}\label{alg:rep}
\small
\Procedure{READ}{$slot$}
  \State $v=\texttt{RDMA\_READ\_primary}(slot)$
  \If {$v = \texttt{FAIL}$} deal with failure
  \EndIf
  \State \Return{$v$}
\EndProcedure
\Procedure{WRITE}{$slot,v_{new}$}
  \State $v_{old} = \texttt{RDMA\_READ\_primary}(slot)$
  \State $v\_list = \texttt{RDMA\_CAS\_backups}(slot, v_{old}, v_{new})$
  \LineComment{\textit{Change all the $v_{old}$s in the $v\_list$ to $v_{new}$s.}}
  \State $v\_list = \texttt{change\_list\_value}(v\_list, v_{old}, v_{new})$
  \State $win=\texttt{EVALUATE\_RULES}(v\_list)$    \Comment{The last writer returns the winning rule while other writers return \texttt{LOSE}.}
  \If {$win=\texttt{Rule\_1}$}
    \State $\texttt{RDMA\_CAS\_primary}(slot, v_{old}, v_{new})$
  \ElsIf {$win \in \{\texttt{Rule\_2}, \texttt{Rule\_3}\}$}
    \State $\texttt{RDMA\_CAS\_backups}(slot, v\_list, v_{new})$
    \State $\texttt{RDMA\_CAS\_primary}(slot, v_{old}, v_{new})$
  \ElsIf {$win = \texttt{LOSE}$}
    \Repeat
      \State sleep a little bit
      \State $v_{check}=\texttt{RDMA\_READ\_primary}(slot)$
      \If {notified failure}
        \textit{goto} Line~24
      \EndIf
    \Until {$v_{check} \neq v_{old}$}
    \If {$v_{check} = \texttt{FAIL}$}
      \textit{goto} Line~24
    \EndIf
  \ElsIf {$win = \texttt{FAIL}$}
    \State deal with failure
  \EndIf
  \State \Return{}
\EndProcedure
\end{algorithmic}
\end{algorithm}
\vspace{-0.4in}
\end{figure}

\begin{figure}
\begin{algorithm}[H]
\begin{algorithmic}[1]
\caption{The rule evaluation procedure of \rep}\label{alg:rep-eval}
\small
\Procedure{evaluate\_rules}{$v\_list, slot, v_{new}, v_{old}$}
  \State $v_{maj} = $ The majority value in $v\_list$
  \State $cnt_{maj} = $ The number of $v_{maj}$ in $v\_list$
  \If {$\texttt{FAIL} \in v\_list$}
    \State \Return{$\texttt{FAIL}$}
  \ElsIf {$cnt_{maj} = \texttt{Len}(v\_list)$}
    \State \Return{$\texttt{Rule 1} \text{ if } v_{maj} = v_{new} \text{ else } \texttt{LOSE}$}
  \ElsIf {$2 * cnt_{maj} > \texttt{Len}(v\_list)$}
    \State \Return{$\texttt{Rule 2} \text{ if } v_{maj} = v_{new} \text{ else } \texttt{LOSE}$}
  \ElsIf {$v_{new} \not\in v\_list $}
    \State \Return{$\texttt{LOSE}$}
  \EndIf
  \State $v_{check} = \texttt{RDMA\_READ}(slot)$
  \If {$v_{check} = \texttt{FAIL}$}
    \State \Return{$\texttt{FAIL}$}
  \ElsIf {$v_{check} \neq v_{old}$}
    \State \Return{$\texttt{FINISH}$}
  \ElsIf {$min(v\_list) = v_{new}$}
    \State \Return{$\texttt{Rule 3}$}
  \EndIf
  \State \Return{$\texttt{LOSE}$}
\EndProcedure
\end{algorithmic}
\end{algorithm}
\vspace{-0.4in}
\end{figure}

%% file: Sections/05-recovery.tex
\section{Failure Handling}\label{sec:execution}
\noindent
Similar to existing replication protocols~\cite{hermes,cr,craq}, \kv relies on a fault-tolerant master with a lease-based membership service~\cite{ukharon} to handle failures.
The master maintains a membership lease for both clients and MNs so that clients always know alive MNs by periodically extending their leases.
The failures of both clients and MNs can be detected by the master when they no longer extend their leases.
Master crashes are handled by replicating the master with state machine replication~\cite{ukharon,craq,cr}.
We formally verify \kv in TLA+~\cite{tla+} for safety and absence of deadlocks under MN failures and the detailed illustration can be found in Appendix~\ref{sec:proof}.

\subsection{Failure Model}
\noindent
We consider a partially synchronous system where processes, {\em i.e.}, clients and MNs, are equipped with loosely synchronized clocks~\cite{partialsynchrony,ukharon,hermes}.
\kv assumes \textit{crash-stop} failures, where processes, {\em i.e.}, clients and MNs, may fail due to crashing and their operations are non-Byzantine.

Under this failure model, \kv guarantees linearizable operations, {\em i.e.}, each KV operation is atomically committed in a time between its invocation and completion~\cite{lin}.
All the objects of \kv are durable and available under an arbitrary number of client crashes and at most $r - 1$ MN crashes, where $r$ is the replication factor.

\subsection{Memory Node Crashes}
\noindent
MN crashes lead to failed accesses to KV pairs and hash slots.
For accesses to KV pairs, clients can access the backup replicas according to the consistent hashing scheme.

The complication comes from the unavailable primary and backup slots that affect the normal execution of index \texttt{READ} and \texttt{WRITE} operations.
\kv relies on the fault-tolerant master to execute operations on clients' behalves under MN failures.
We first introduce how clients \texttt{READ}/\texttt{WRITE} the replicated slots and then introduce the master's operations.

When executing index \texttt{WRITE} under MN crashes, \kv allows the last writer decided by the \rep replication protocol to continue modifying all \emph{alive} slots to the same value.
Other writers send RPC requests to the master and wait for the master to reply with a correct value in the replicated slots.
Under situations when no last writer can be decided, the master decides the last writer and modifies all the index slots on behalf of clients.
For \texttt{READ} operations, executions are not affected under the following two cases.
First, if the primary slot is still alive, clients can read the primary slot normally.
Second, if the primary slot crashes, clients read all \textit{alive} backup slots. 
If all \textit{alive} backup slots contain the same value, reading this value is safe because there are no write conflicts.
Otherwise, clients use RPCs and rely on the master to return a correct value for the crashed slot.
Since \texttt{READ} operations are only affected under write conflicts, most \texttt{READ} can continue under the read-intensive workloads that dominate in real-world situations~\cite{facebookKVBenchmark,twitterwl}.

On detecting MN crashes, the master first blocks clients from further modifying the crashed slots with the lease expiration.
The master then acts as a representative last writer that modifies all \emph{alive} slots to the same value.
Specifically, the master selects a value $v$ in an \emph{alive} backup slot and modifies all \emph{alive} slots to $v$.
Since the \rep protocol modifies the backup slots before the primary slot, the values in the backup slots are always newer than the primary slot.
Hence, the master choosing a value from a backup slot is correct because it proceeds the conflicting write operations.
In cases where all backup slots crash, the master selects the value in the primary slot.
Clients that receive old values from the master retry their write operations to guarantee that their new value is written.
The master then writes the \textit{old value} in the operation log header to prevent clients from redoing operations when recovering from crashed clients (Section~\ref{sec:client-crash}).
Finally, the master reconfigures new primary and backup slots and returns the selected value to all clients that wait for a reply.
After the reconfiguration of the primary and backup slots, all KV requests can be executed normally without involving the master.
During the whole process, only accesses to the crashed slots are affected and the blocking time can be short thanks to the microsecond-scale membership service~\cite{ukharon}.

\subsection{Client Crashes}\label{sec:client-crash}
\noindent
Crashed clients may result in two issues.
First, their allocated memory blocks remain unmanaged, causing memory leakage.
Second, other clients may be unable to modify a replicated index slot if the crashed client is the last writer.
The master uses embedded operation logs to address these two issues.

The recovery process is executed in the compute pool and consists of two steps, {\em i.e.}, memory re-management and index repair.
Memory re-management restores the coarse-grained memory blocks allocated by the client and the fine-grained object usage information of the client.
The recovery process first gets all memory blocks managed by the crashed client by letting MNs search for their local block allocation tables.
Then the recovery process traverses the per-size-class linked lists to find all used objects and log entries.
With the used objects and the allocated memory blocks, the recovery process can easily restore the free object lists of the crashed client.
Hence, all the memory spaces of the crashed client are re-managed.

The index repair procedure then fixes the partially modified hash index.
\kv deems all requests at the end of per-size-class linked lists as potentially crashed requests.
For incomplete log entries, {\em i.e.}, the \textit{used bit} at the end of the log entry is not set, the client must crashes during writing the KV pair (c0 in Figure~\ref{fig:kv-ops}).
The object is directly reclaimed without further operation since the writing of the object has not been completed.
For a log entry with an incomplete \textit{old value} according to the \textit{CRC} field, \kv redoes the request according to the \textit{operation field} and the KV pair.
Under this situation, either the request belongs to the last writer that crashed before committing the log (c1 in Figure~\ref{fig:kv-ops}), or it belongs to other non-last writers.
In the first case, the values in the backup slots may not be consistent and the primary slot has not been modified to a new value.
Redoing the request can make the backup and primary slots consistent.
In the second case, since the request of crashed non-last writers has not returned to clients, redoing the request does not violate the linearizability.
For a request with a complete \textit{old value}, the request must belong to a last writer.
However, the request may finish (c3) or crash before the primary slot is modified (c2).
The recovery process checks the value in the primary slot ($v_p$) and the value in the \textit{old value} ($v_{old}$) to distinguish c2 from c3.
If $v_p = v_{old}$, the request crashed before the primary was modified because $v_{old}$ records the value before index modification.
Since all backup slots are consistent, the recovery process modifies the primary slot to the new value and finishes the recovery.
Otherwise, the request is finished and no further operation is required.
After recovering the request, the master asynchronously checks content in the $v_{old}$s in log entries of the crashed client to recover its batched free operations.

\subsection{Mixed Crashes}
\noindent
In situations where clients and MNs crash together, \kv recovers the failures separately.
\kv first lets the master recover all MN crashes and then starts the recovery processes for failed clients.
KV requests can proceed because the master acts as the last writer for all blocked KV requests.
No request is committed twice because the master commits the operation logs on clients' behalves.

%% file: Sections/06-evaluation.tex
\begin{figure*}[t]
    \vspace{-0.1in}
    \centering
    \subfloat[\texttt{INSERT} latency CDF.]{
        \includegraphics[width=0.47\columnwidth]{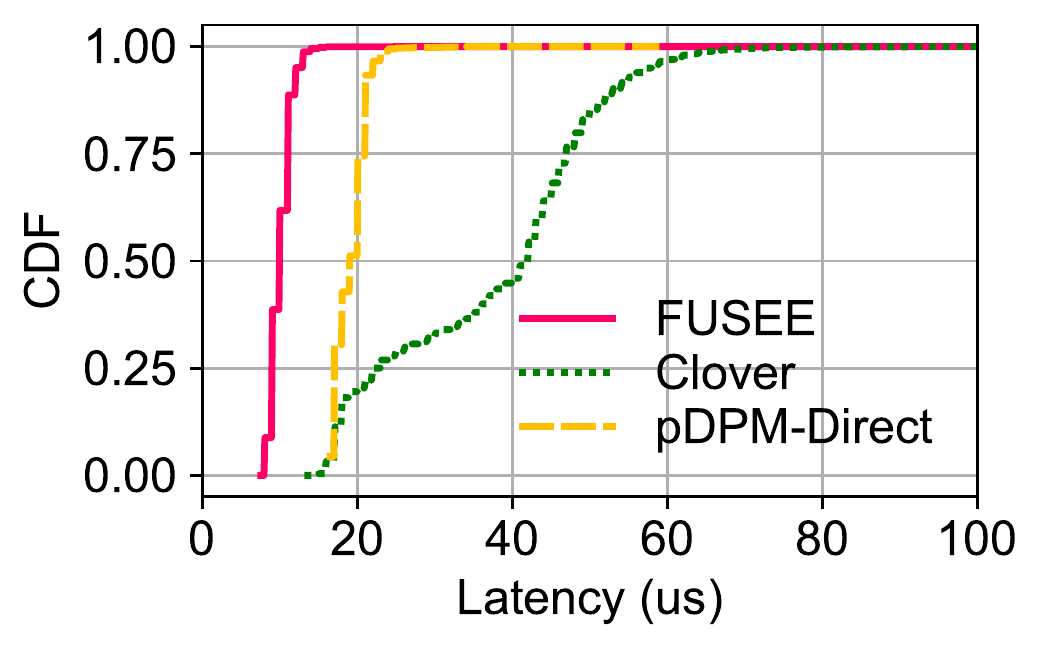}
        \label{fig:insert-lat-cdf}
    }
    \subfloat[\texttt{UPDATE} latency CDF.]{
        \includegraphics[width=0.47\columnwidth]{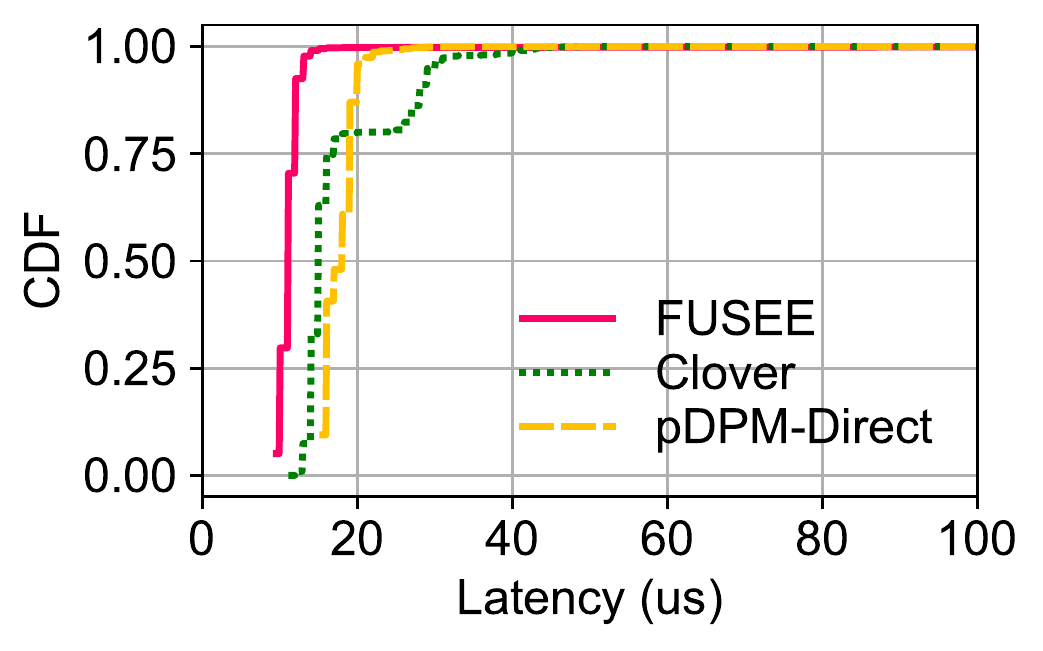}
        \label{fig:update-lat-cdf}
    }
    \subfloat[\texttt{SEARCH} latency CDF.]{
        \includegraphics[width=0.47\columnwidth]{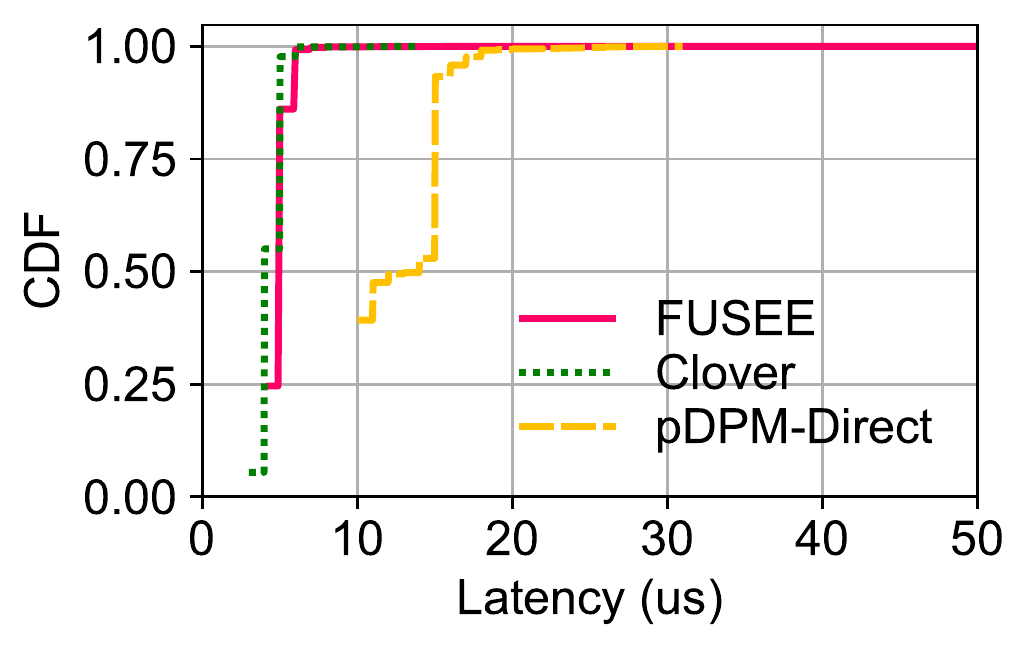}
        \label{fig:search-lat-cdf}
    }
    \subfloat[\texttt{DELETE} latency CDF.]{
        \includegraphics[width=0.47\columnwidth]{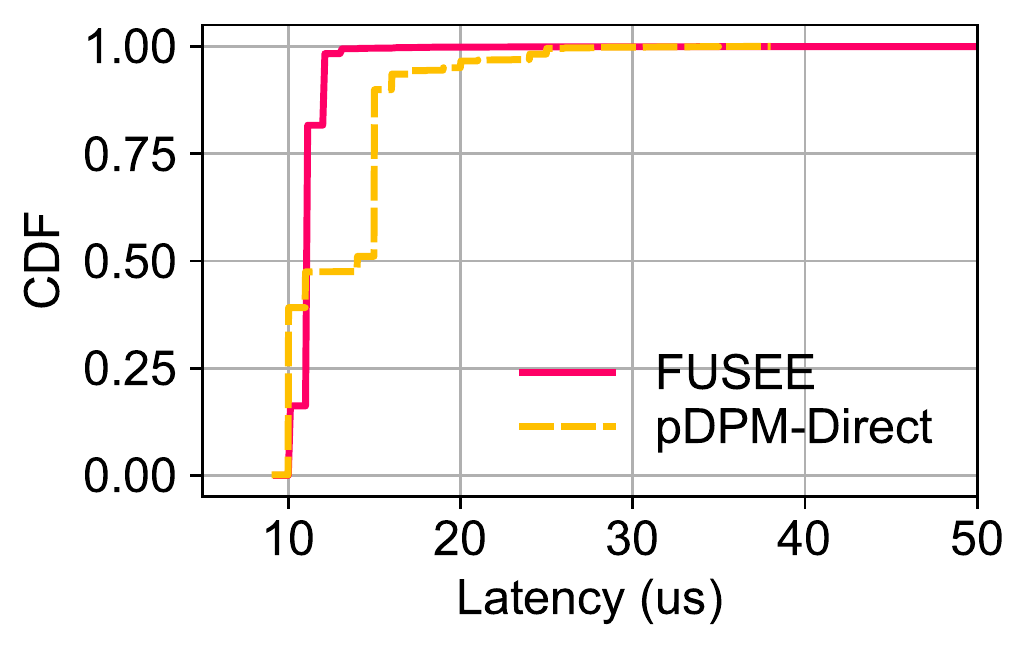}
        \label{fig:delete-lat-cdf}
    }
    \caption{The CDFs of different KV request latency under the microbenchmark.}
    \label{fig:lat-cdf}
    \vspace{-0.2in}
\end{figure*}

\section{Evaluation}\label{sec:eval}
\subsection{Experiment Setup}
\noindent
\textbf{Implementation.}
We implement \kv from scratch in C++ with 13k LOC.
We implement RACE hashing carefully according to the paper due to no available open-source implementations.
Coroutines are employed on clients to hide the RDMA polling overhead, as suggested in~\cite{fasst,racehash}.
The design of \kv is agnostic to the lower-level memory media of memory nodes, {\em i.e.}, any memory node with either persistent memory (PM) or DRAM that provides \texttt{READ}, \texttt{WRITE}, and 8-byte \texttt{CAS} interfaces is compatible.
We adopt monolithic servers with RNICs and DRAM to serve as MNs like Clover~\cite{pdpm} since we do not have access to smartNICs and PM.
Specifically, we start an MN process on a monolithic server to register RDMA memory regions and serve memory allocation RPCs with a UDP socket.
MN processes serve memory allocation requests with UDP sockets.
Since the socket \textit{receive} is a blocking system call, the process will be in the blocked state with no CPU usage most of the time.

\noindent
\textbf{Testbed.}
We run all experiments on 22 physical machines (5 MNs and 17 CNs) on the APT cluster of CloudLab~\cite{cloudlab}.
Each machine is equipped with an 8-core Intel Xeon E5-2450 processor, 16GB DRAM, and a 56Gbps Mellanox ConnectX-3 IB RNIC.
These machines are interconnected with 56Gbps Mellanox SX6036G switches.

\noindent
\textbf{Comparison.}
We compare \kv with two state-of-the-art KV stores on DM, {\em i.e.}, pDPM-Direct and Clover~\cite{pdpm}.
pDPM-Direct stores and manages the KV index and memory space on the clients.
It uses a distributed consensus protocol to ensure metadata consistency and locks to resolve data access conflicts.
We extend the open-source version of pDPM-Direct to support string keys for fair comparison in our evaluation.
Clover is a semi-disaggregated KV store that adopts monolithic servers to manage memory spaces and a hash index.
All \texttt{UPDATE} and \texttt{INSERT} requests have to go through the metadata server, requiring additional compute power.
For both pDPM-Direct and Clover, client-side caches are enabled following their default settings.
To show the effectiveness of \rep and the adaptive index cache, we implement \kv-CR and \kv-NC, two alternative versions of \kv.
\kv-CR replicates index modifications by sequentially \texttt{CAS}ing all replicas to enforce sequential accesses.
\kv-NC is the version of \kv without a client-side cache.
For all these methods, we evaluate their throughput and latency with both micro and YCSB~\cite{ycsb} benchmarks.

Since the open-source version of Clover and pDPM-Direct only support one index replica, we compare \kv with these two approaches with a single index replica and two data replicas in the microbenchmark (Section~\ref{sec:microbenchmark}) and YCSB performance (Section~\ref{sec:ycsb}) evaluations.
When evaluating \kv with a single index replica, the embedded log is constructed, but the commit of the log is skipped since committing the log is used to ensure the consistency of multiple index replicas.
The performance of \kv with multiple replicas is evaluated in the fault-tolerance evaluation (Section~\ref{sec:ft}).

\subsection{Microbenchmark Performance}\label{sec:microbenchmark}
\noindent
We use microbenchmarks to evaluate the operation throughput and latency of the three approaches.
For \kv and pDPM-Direct, we use 16 CNs and 2 MNs.
For Clover, we use 17 CNs and 2 MNs because it needs an additional metadata server, consuming 8 more CPU cores and an additional RNIC. 
We do not use multiple metadata servers for Clover because the current open-source implementation of Clover only supports a single metadata server.
We run 128 client processes on the 16 CNs, where each CN holds 8 clients.
The \texttt{DELETE} of Clover is not tested because Clover does not support it.

\begin{figure}[t]
    \hspace{3mm}
    \begin{minipage}[t]{0.45\columnwidth}
        \centering
        \includegraphics[width=\columnwidth]{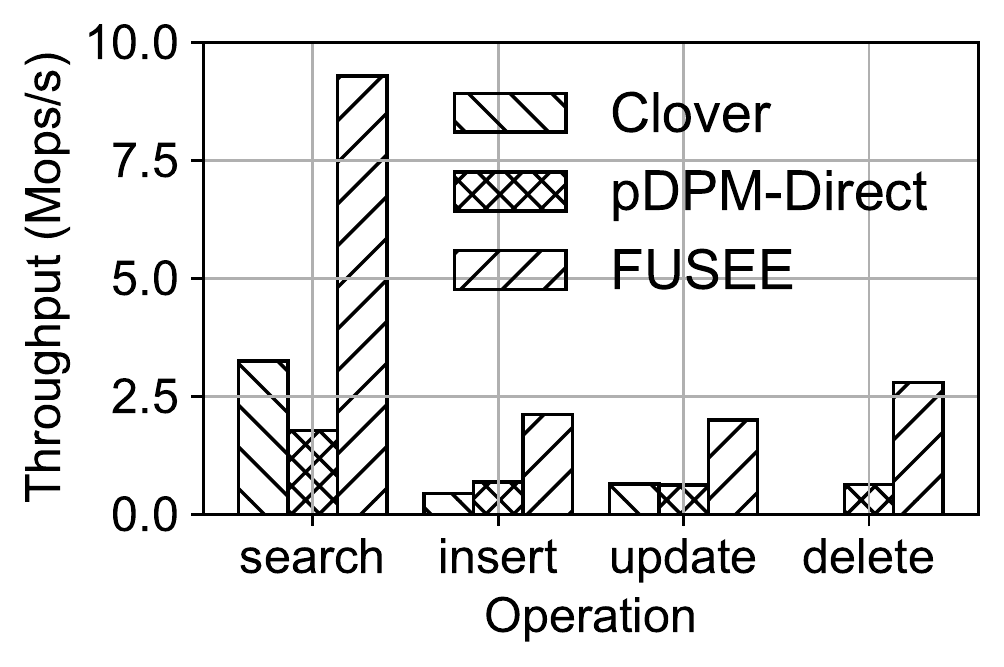}
        \caption{\footnotesize The throughputs of microbenchmark.}
        \label{fig:micro-tpt}
    \end{minipage}%
    \hspace{3mm}
    \begin{minipage}[t]{0.44\columnwidth}
        \centering
        \includegraphics[width=0.99\columnwidth]{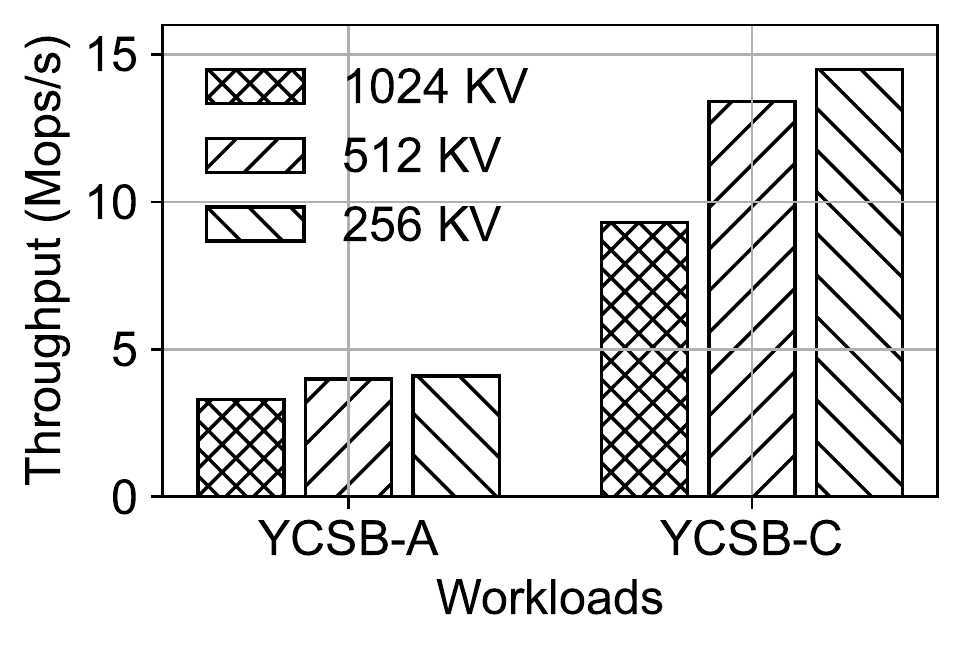}
        \caption{\footnotesize The throughput of \kv under different KV sizes.}
        \label{fig:kv-size}
    \end{minipage}
    \vspace{-0.25in}
\end{figure}

\textbf{Latency.}
To evaluate the latency of KV requests, we use a single client to iteratively execute each operation $10,000$ times.
Figure~\ref{fig:lat-cdf} shows the cumulative distribution functions (CDFs) of the request latency.
\kv performs the best on \texttt{INSERT} and \texttt{UPDATE}, since the \rep replication protocol has bounded RTTs.
\kv has a little higher \texttt{SEARCH} latency than Clover since \kv reads the hash index and the KV pair in a single RTT, which is slower than only reading the KV pair in Clover.
\kv has slightly higher \texttt{DELETE} latency than pDPM-Direct because \kv writes a log entry and reads the hash index in a single RTT, which is slower than just reading the hash index in pDPM-Direct.

\textbf{Throughput.}
Figure~\ref{fig:micro-tpt} shows the throughput of the three approaches.
The throughput of pDPM-Direct is limited by its remote lock, which causes extensive lock contention as the number of clients grows. 
For Clover, even though it consumes more hardware resources, {\em i.e.}, 8 additional CPU cores and an RNIC, the scalability is still lower than \kv.
This is because the CPU processing power of the metadata server bottlenecks its throughput.
On the contrary, \kv improves the overall throughput by eliminating the computation bottleneck of the metadata server and efficiently resolving conflicts with the \rep replication protocol.

\begin{figure*}
    \vspace{-0.2in}
    \centering
    \subfloat[A (\texttt{SEARCH}:\texttt{UPDATE} = 0.5:0.5).]{
        \includegraphics[width=0.48\columnwidth]{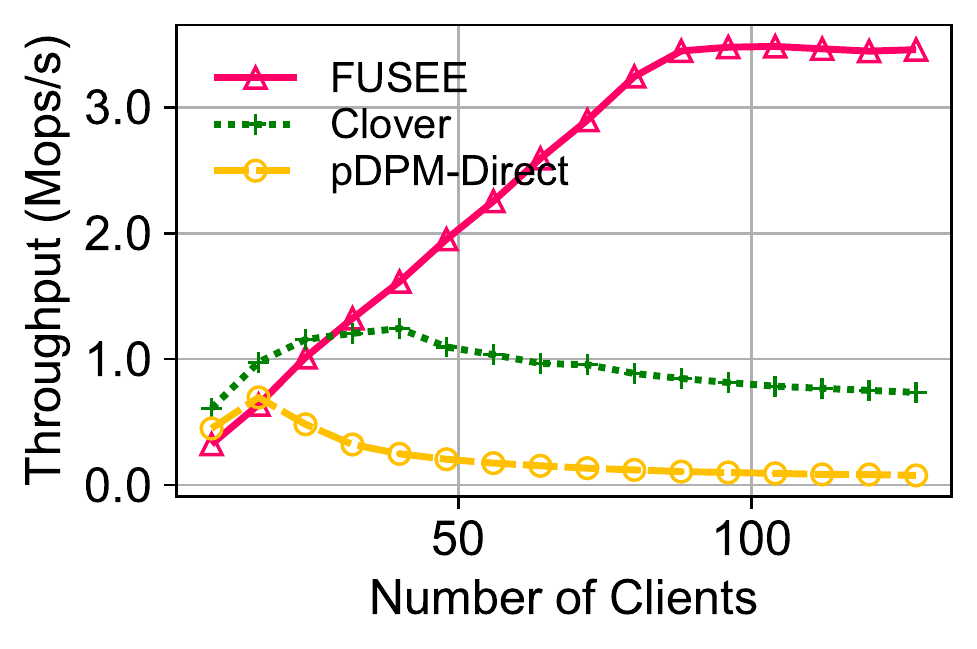}
        \label{fig:ycsb-tpt-a}
    }
    \subfloat[B (\texttt{SEARCH}:\texttt{UPDATE} = 0.95:0.05).]{
        \includegraphics[width=0.48\columnwidth]{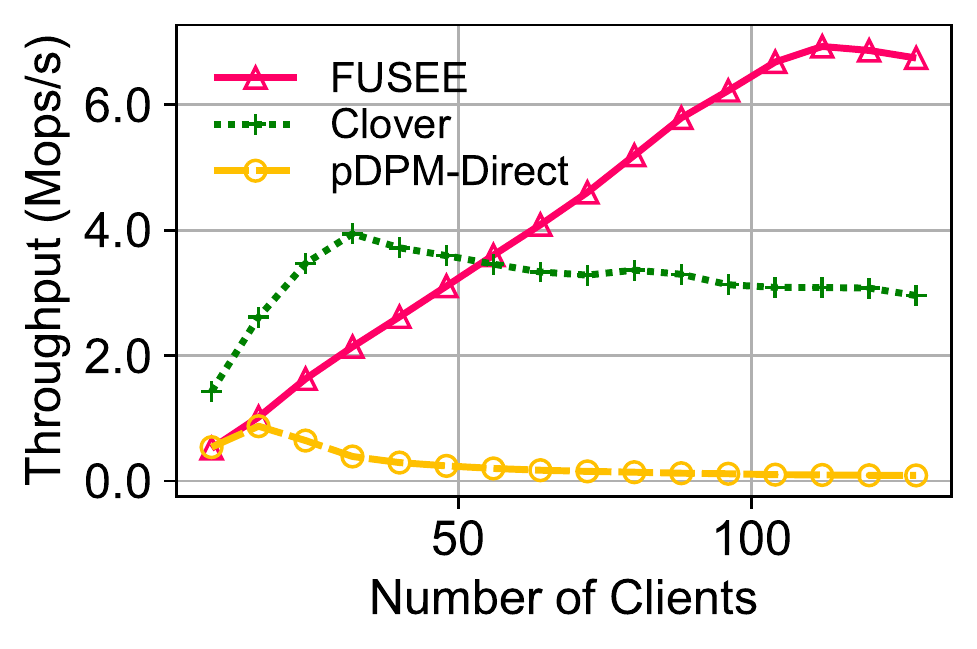}
        \label{fig:ycsb-tpt-b}
    }
    \subfloat[C (100\% \texttt{SEARCH}).]{
        \includegraphics[width=0.48\columnwidth]{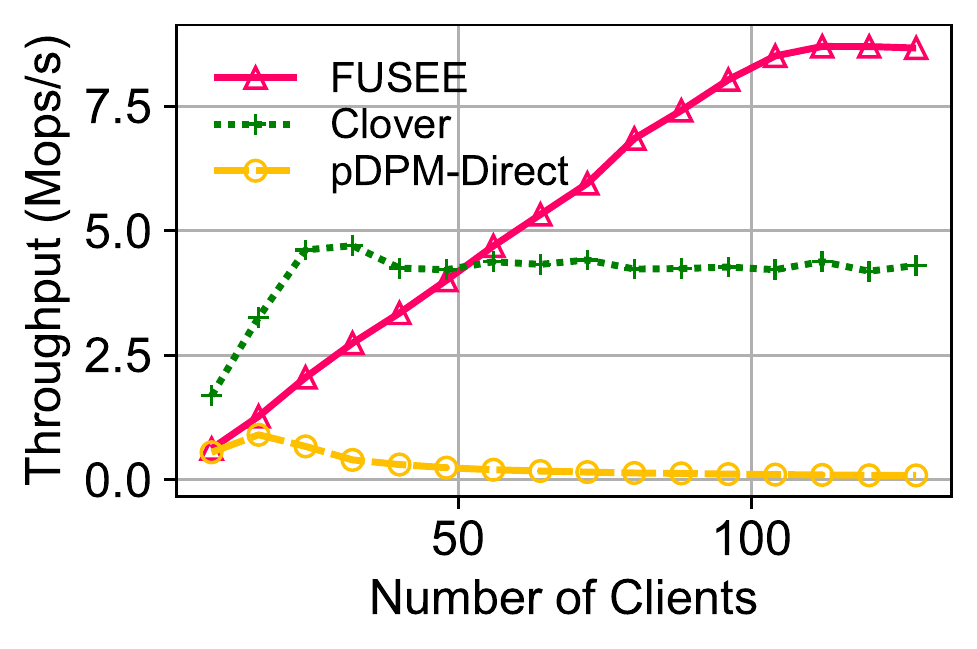}
        \label{fig:ycsb-tpt-c}
    }
    \subfloat[D (\texttt{SEARCH}:\texttt{INSERT} = 0.95:0.05).]{
        \includegraphics[width=0.48\columnwidth]{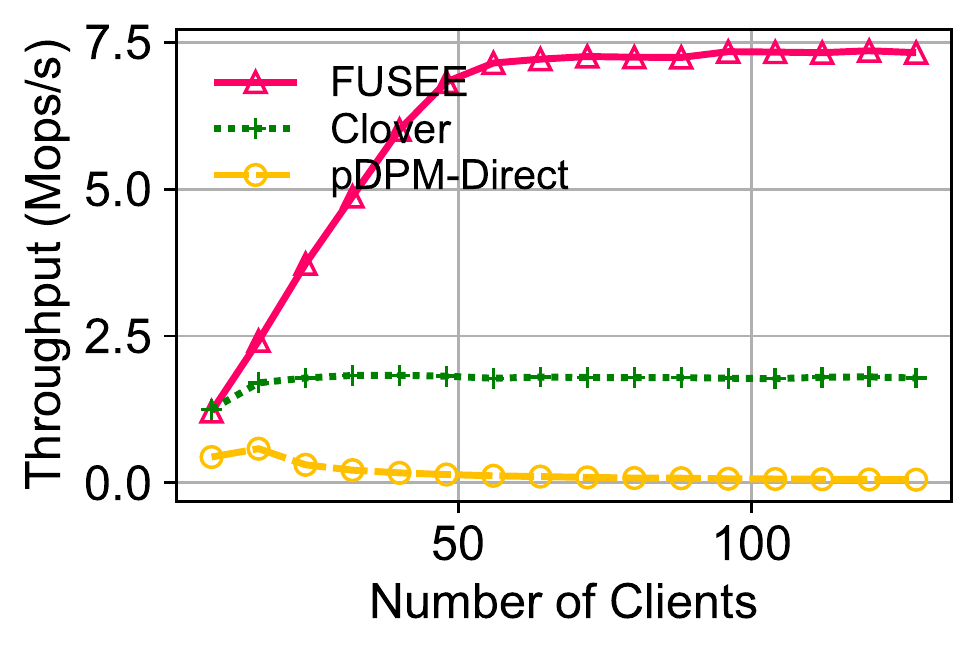}
        \label{fig:ycsb-tpt-d}
    }
    \caption{The scalability of \kv under different YCSB workloads.}
    \label{fig:ycsb-scalability}
    \vspace{-0.25in}
\end{figure*}

\begin{figure}[t]
    \vspace{-0.15in}
    \centering
    \subfloat[YCSB-A throughput.]{
        \includegraphics[width=0.45\columnwidth]{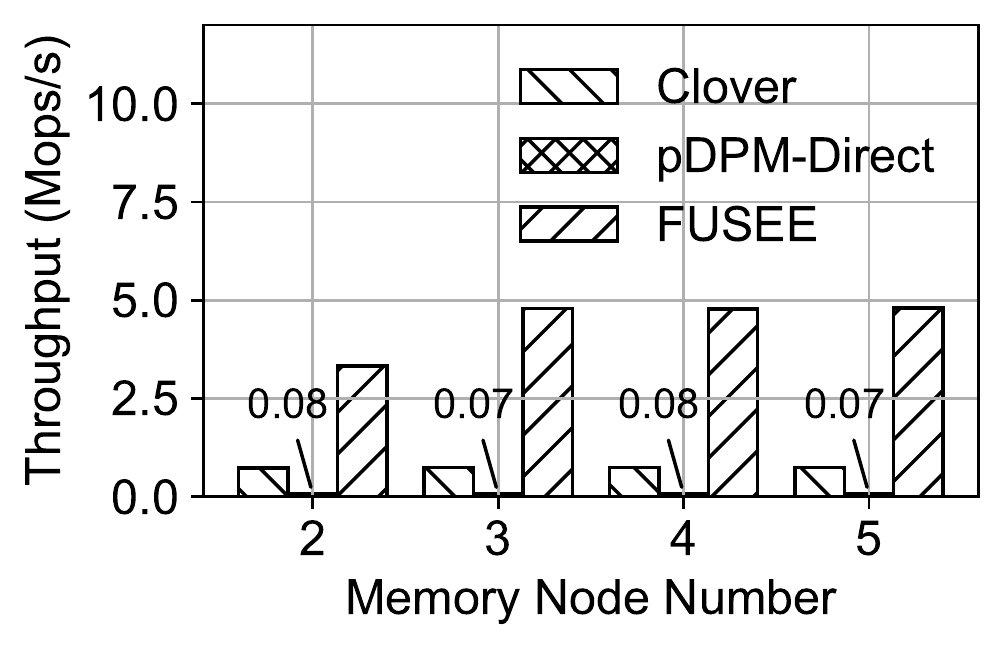}
        \label{fig:mn-ycsb-a-tpt}
    }
    \subfloat[YCSB-C throughput.]{
        \includegraphics[width=0.45\columnwidth]{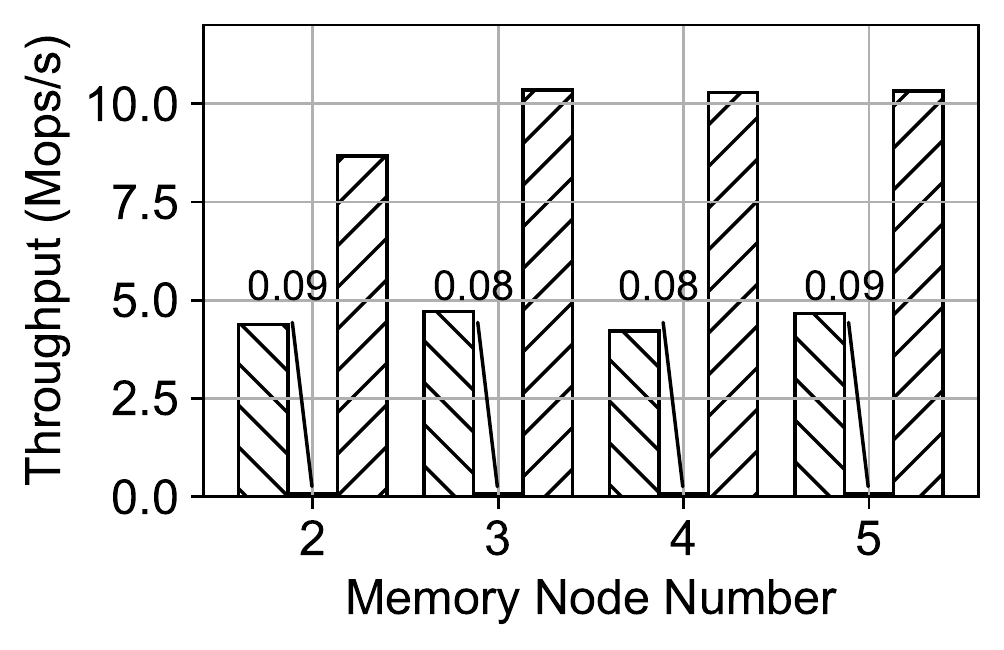}
        \label{fig:mn-ycsb-c-tpt}
    }
    \caption{\small The throughput with different numbers of MNs.}
    \label{fig:mn-tpt}
    \vspace{-0.2in}
\end{figure}

\begin{figure}[t]
    \hspace{2mm}
    \begin{minipage}[t]{.45\columnwidth}
        \centering
        \includegraphics[width=\columnwidth]{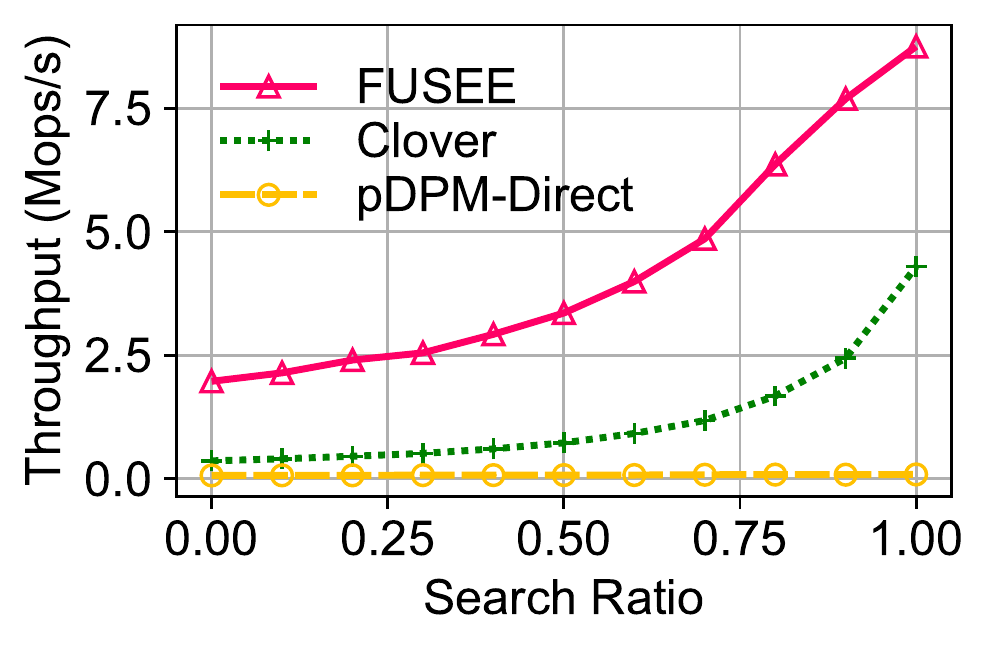}
        \caption{\footnotesize Throughput under different \texttt{SEARCH-UPDATE} ratios.}
        \label{fig:read-write-tpt}
    \end{minipage}%
    \hspace{2mm}
    \begin{minipage}[t]{.45\columnwidth}
        \centering
        \includegraphics[width=\columnwidth]{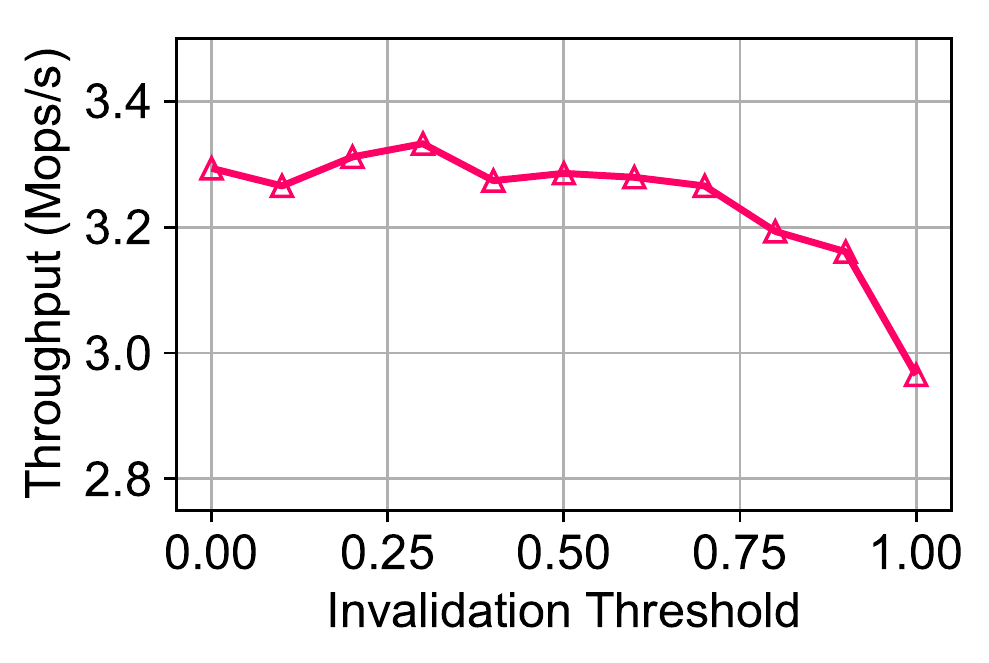}
        \caption{\footnotesize Throughput under different adaptive cache thresholds.}
        \label{fig:cache-miss-tpt}
    \end{minipage}
    \vspace{-0.25in}
\end{figure}

\subsection{YCSB Performance}\label{sec:ycsb}
\noindent
For YCSB benchmarks~\cite{ycsb}, we generate $100,000$ keys with the Zipfian distribution ($\theta=0.99$).
We use $1024$-byte KV pairs, which is representative of real-world workloads~\cite{ycsb,facebookKVBenchmark,rocksDBExperience}.
The hardware setup is the same as microbenchmarks.

\textbf{YCSB Throughput.}
Figure~\ref{fig:ycsb-scalability} shows the throughput of three approaches with different numbers of clients.
Clover performs the best under a small number of clients since adopting the metadata server simplifies KV operations.
Compared with Clover, pDPM-Direct and \kv require more RDMA operations to resolve index modification conflicts.
As the number of clients grows, the throughput of Clover and pDPM-Direct does not increase because the throughput is bottlenecked by the metadata server and the lock contention, respectively.
Compared with Clover, \kv scales better with the growing number of clients while consuming fewer resources.
Compared with pDPM-Direct, \kv improves the throughput by avoiding lock contention.
When the number of clients reaches 128, the throughput of \kv is $4.9\times$ and $117\times$ higher than Clover and pDPM-Direct, respectively.

Figure~\ref{fig:mn-tpt} shows the throughput of the three approaches with a write-intensive workload (YCSB-A) and a read-intensive workload (YCSB-C) when varying numbers of MNs from 2 to 5 using 128 clients. 
The throughput of pDPM-Direct and Clover does not increase due to being limited by lock contention and the limited compute power of the metadata server, respectively.
As for \kv, the throughput improves as the number of memory nodes increases from 2 to 3.
There is no further throughput improvement because the total throughput is limited by the number of compute nodes.

\begin{figure}[t]
    \vspace{-0.03in}
    \begin{minipage}[t]{.46\columnwidth}
        \centering
        \includegraphics[width=0.99\columnwidth]{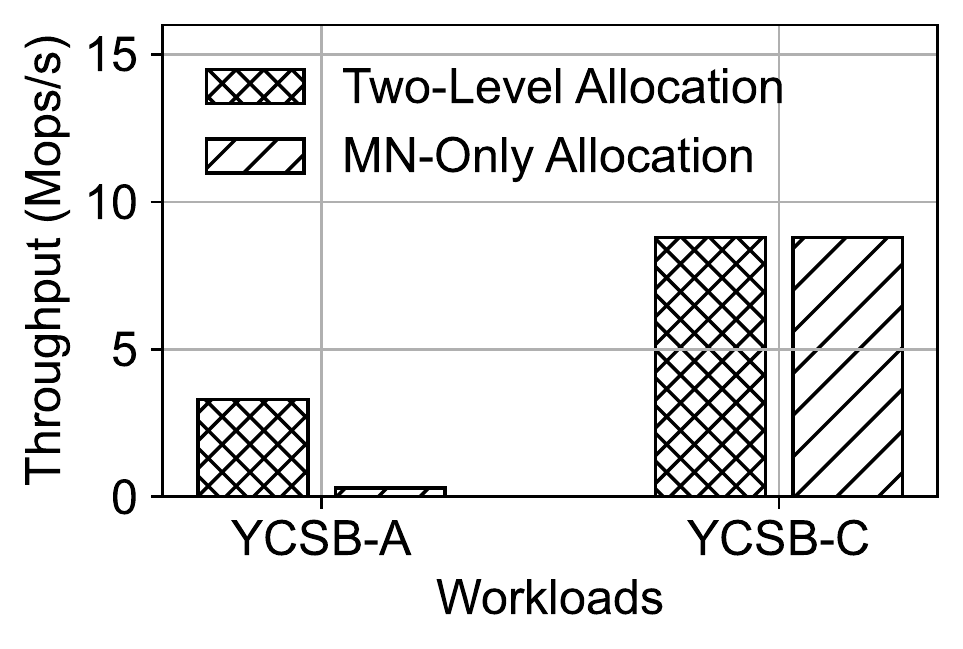}
        \caption{\footnotesize The throughput of different memory allocation methods.}
        \label{fig:mm-tpt}
    \end{minipage}%
    \hspace{2mm}
    \begin{minipage}[t]{.45\columnwidth}
        \centering
        \includegraphics[width=\columnwidth]{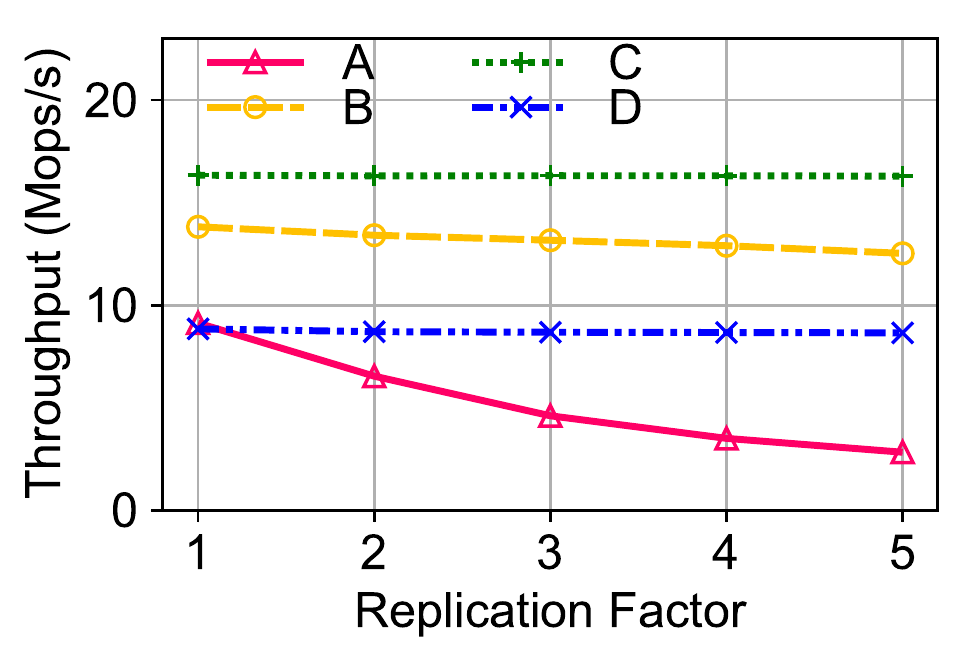}
        \caption{\footnotesize YCSB throughput under different replication factors.}
        \label{fig:rep-tpt}
    \end{minipage}
    \vspace{-0.25in}
\end{figure}

\begin{figure*}
    \centering
    \subfloat[\texttt{UPDATE} median latency.]{
        \includegraphics[width=0.46\columnwidth]{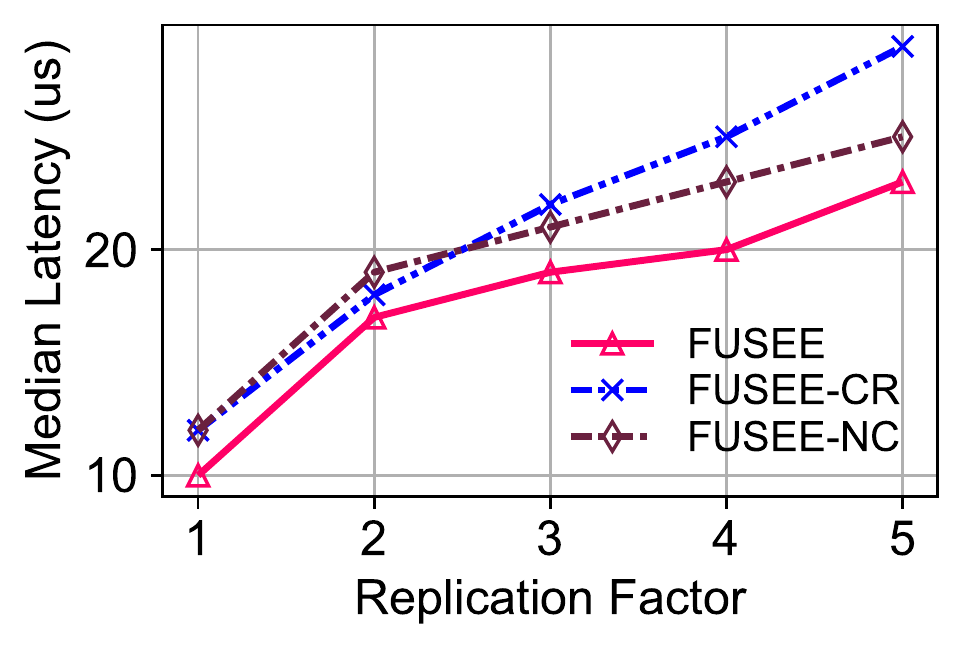}
        \label{fig:rep-update-lat}
    }
    \subfloat[\texttt{DELETE} median latency.]{
        \includegraphics[width=0.46\columnwidth]{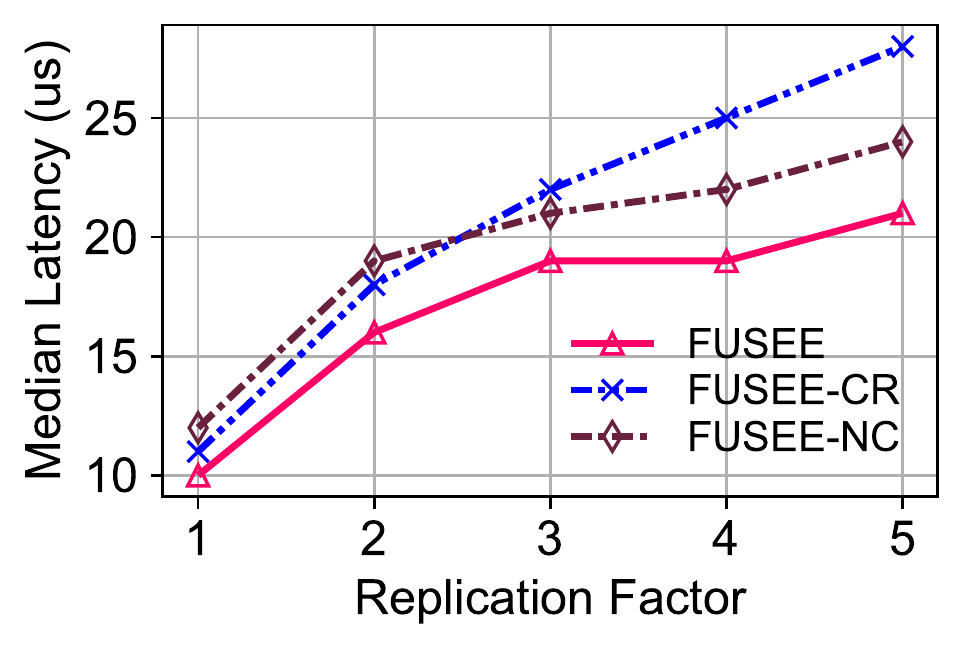}
        \label{fig:rep-delete-lat}
    }
    \subfloat[\texttt{INSERT} median latency.]{
        \includegraphics[width=0.46\columnwidth]{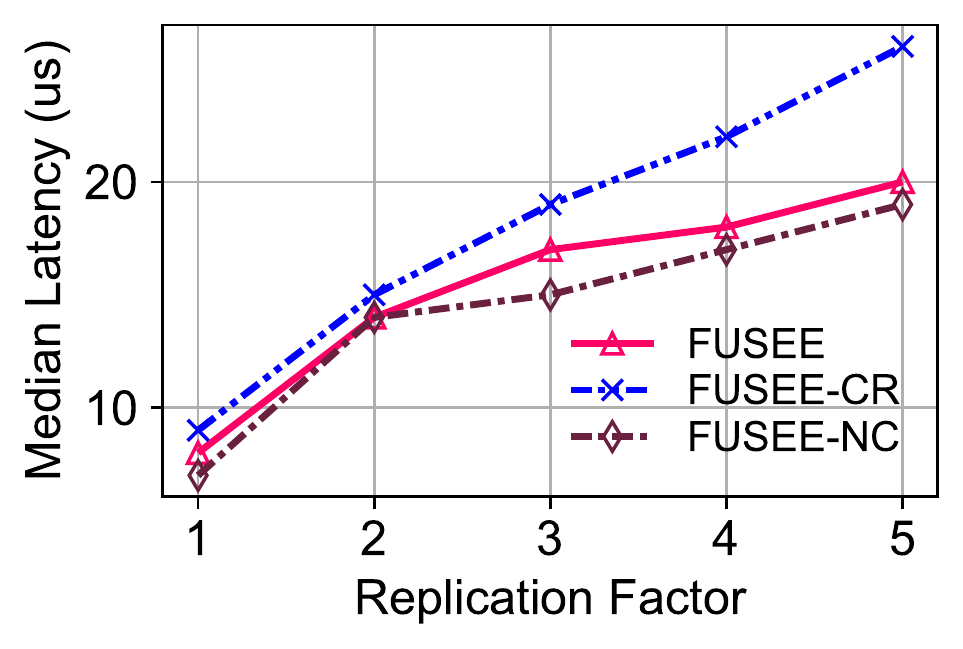}
        \label{fig:rep-insert-lat}
    }
    \subfloat[\texttt{SEARCH} median latency.]{
        \includegraphics[width=0.45\columnwidth]{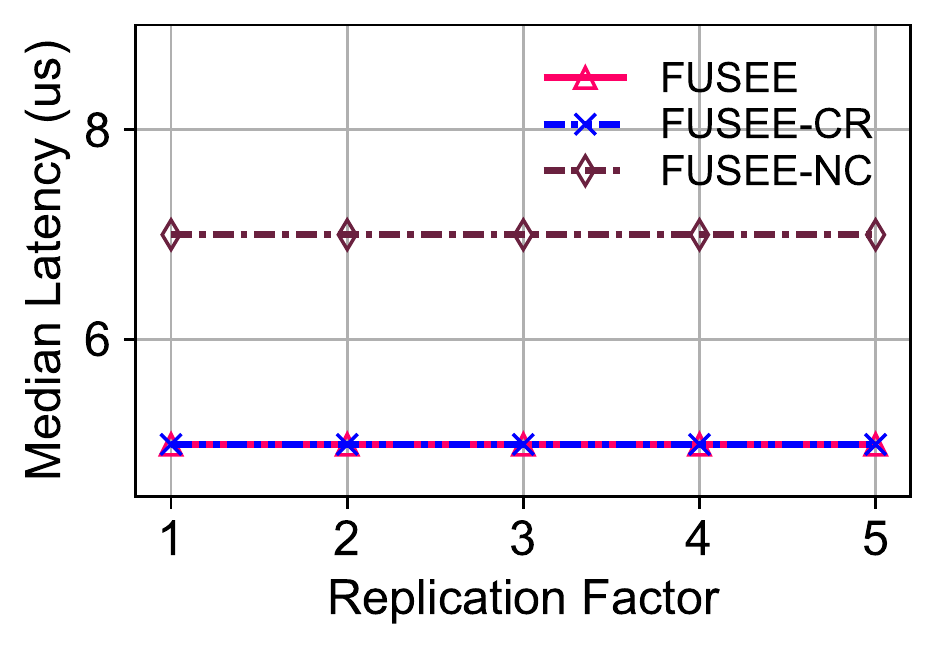}
        \label{fig:rep-search-lat}
    }
    \caption{Median operation latency of \kv, \kv-NC and \kv-CR under different replication factors.}
    \label{fig:rep-lat}
    \vspace{-0.2in}
\end{figure*}

Figure~\ref{fig:kv-size} shows the throughput of \kv under smaller KV sizes.
Since the throughput of \kv is limited by the bandwidth of MN-side RNICs, the YCSB-C throughput of \kv increases by $44.1\%$ and $55.9\%$ with 512B and 256B KV pairs, respectively.
The performance of \kv is not affected by the dataset size because the performance depends only on the number of RTTs of KV requests, which is deterministic as presented in Section~\ref{sec:design}.

\textbf{Read-write performance.}
Figure~\ref{fig:read-write-tpt} shows the throughput of the three approaches under different \texttt{SEARCH}-\texttt{UPDATE} ratios.
As the portion of \texttt{UPDATE} grows, the throughput of all three methods decreases because \texttt{UPDATE} requests involve more RTTs.
However, \kv exhibits the best throughput due to eliminating the computation bottleneck of metadata servers.

\textbf{Adaptive index cache performance.}
Figure~\ref{fig:cache-miss-tpt} shows the YCSB-A throughput of \kv with different adaptive index cache thresholds.
The throughput of \kv decreases with the increasing thresholds because more bandwidth is wasted on fetching invalidated KV pairs with a high threshold.

\textbf{Two-level memory allocation performance.}
To show the effectiveness of the two-level memory allocation scheme, we compare \kv with an MN-centric memory allocation scheme, as shown in Figure~\ref{fig:mm-tpt}.
The YCSB-A throughput drops $90.9\%$ due to the limited compute power on MNs, while the YCSB-C throughput remains the same since no memory allocation is involved in the read-only setting.

\begin{figure}[t]
    \hspace{4mm}
    \begin{minipage}[t]{.47\columnwidth}
        \centering
        \includegraphics[width=0.99\columnwidth]{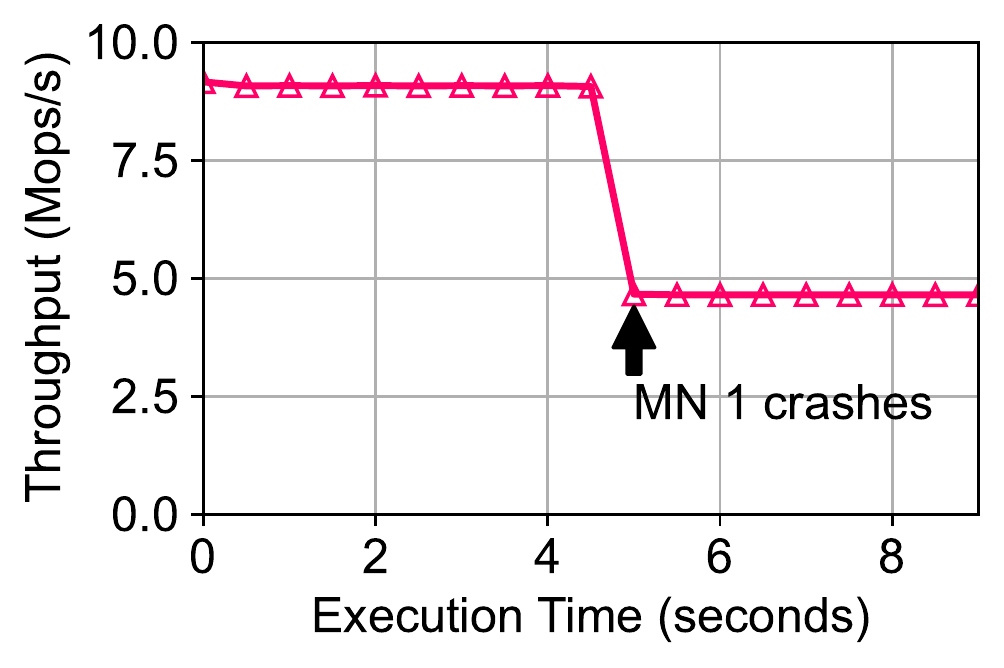}
        \caption{\small YCSB-C throughput under a crashed memory node.}
        \label{fig:crash-tpt}
    \end{minipage}%
    \hspace{2mm}
    \begin{minipage}[t]{.44\columnwidth}
        \centering
        \includegraphics[width=0.99\columnwidth]{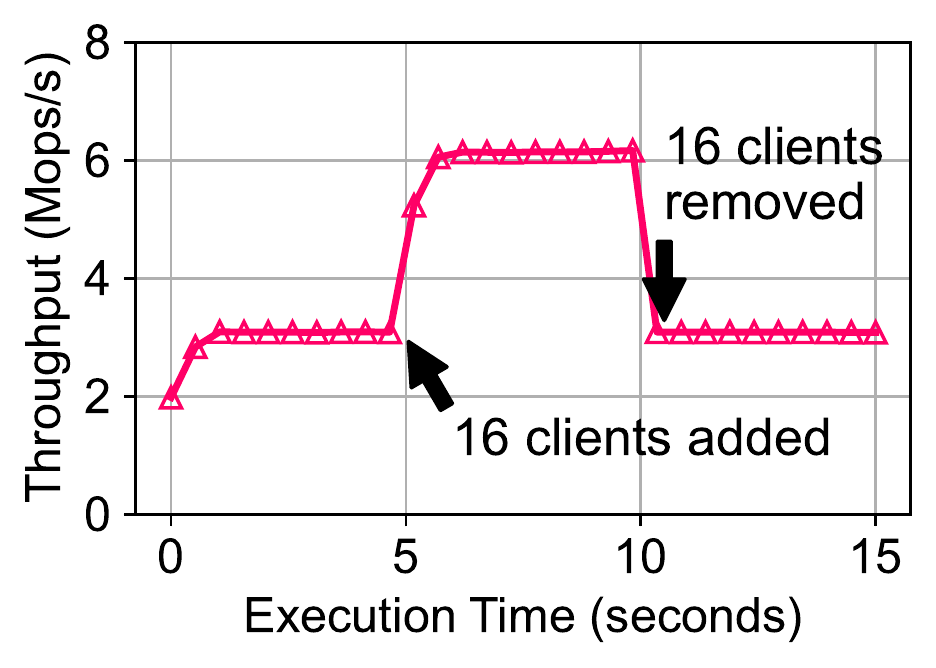}
        \caption{\small The elasticity of \kv.}
        \label{fig:cont-tpt}
    \end{minipage}
    \vspace{-0.15in}
\end{figure}

\subsection{Fault Tolerance \& Elasticity}\label{sec:ft}
\noindent
\textbf{\rep Replication Protocol.}
Figure~\ref{fig:rep-lat} shows the median latency of \kv, \kv-NC, and \kv-CR with different replication factors under microbenchmarks. We set both the numbers of index replicas and data replicas to $r$ where $r$ is the replication factor.
The latency of \kv-CR on \texttt{INSERT}, \texttt{UPDATE}, and \texttt{DELETE} grows linearly as the replication factor because it modifies index replicas sequentially, and the number of RTTs equals the replication factor.
Differently, the latency of \kv grows slightly with the replication factor because \rep has a bounded number of RTTs.
For \texttt{SEARCH} requests, \kv and \kv-CR have comparable latency since they execute \texttt{SEARCH} similarly.
Compared with \kv-NC, \kv has lower latency for \texttt{UPDATE}, \texttt{DELETE}, and \texttt{SEARCH} due to fewer RTTs.
The \texttt{INSERT} latency is slightly higher than that of \kv-NC because \kv spends additional time to maintain the local cache.
Figure~\ref{fig:rep-tpt} shows the throughput of \kv under different replication factors.
For YCSB-A and YCSB-B, the throughput drops as the replication factor grows.
The YCSB-D throughput slightly drops from 8.8 Mops to 8.6 Mops due to the read-intensive nature of YCSB-D.
The YCSB-C throughput remains the same due to no index modifications.

\textbf{Search under Crashed MNs.}
\kv allows \texttt{SEARCH} requests to continue when MNs crash under read-intensive workloads.
Figure~\ref{fig:crash-tpt} shows the throughput of 9 seconds of execution, where memory node 1 crashes at the 5th second.
The overall throughput drops to half of the peak throughput because all data accesses come to one MN.
The throughput is then limited by the network bandwidth of a single RNIC.


\textbf{Recover from Crashed Clients.}
To evaluate the efficiency of a client recovering from failures, we crash and recover a client after \texttt{UPDATE} $1,000$ times.
As shown in Table~\ref{tab:client-rec-bd}, \kv takes 177 milliseconds to recover from a client failure. 
The memory registration and connection re-establishment account for $92\%$ of the total recovery time.
The log traversal and KV request recovery only account for $4\%$ of the recovery time, which implies the affordable overhead of log traversal.

\textbf{Elasticity.}
\kv supports dynamically adding and shrinking clients.
We show the elasticity of \kv by dynamically adding and removing 16 clients when running the YCSB-C workload.
As shown in Figure~\ref{fig:cont-tpt}, the throughput increases when the number of clients increases from 16 to 32 and resumes to the previous level after removing 16 clients.

\begin{table}[t]
    \footnotesize
    \centering
    \caption{Client recovery time breakdown.}
    \label{tab:client-rec-bd}
    \begin{tabular}[width=\columnwidth]{l|r|r}
        \toprule[2pt]
        \textbf{Step} & \textbf{Time (ms)} & \textbf{Percentage} \\
        \midrule[1pt]
        Recover connection \& MR & 163.1 & 92.1\% \\
        Get Metadata           & 0.3   & 0.2\%  \\
        Traverse Log           & 3.5   & 2.0\%  \\
        Recover KV Requests    & 3.5   & 2.0\%  \\
        Construct Free List    & 6.6   & 3.7\%  \\
        \midrule[1pt]
        \textbf{Total}         & 177.0 & 100\%  \\
        \bottomrule[2pt]
    \end{tabular}
    \vspace{-0.1in}
\end{table}

%% file: Sections/07-literature.tex
\section{Related Work}\label{sec:liter}
\noindent
\textbf{Disaggregated Memory.}
Existing approaches can be classified into software-based, hardware-based, and co-design-based memory disaggregation.
Software-based approaches hide the DM abstraction by modifying operating systems~\cite{legoos,lite,infiniswap,canfarmem,zombieland}, virtual machine monitors~\cite{syslevelimpl}, or runtimes~\cite{aifm,semeru}.
Hardware-based ones design memory buses~\cite{disaggmem,genz} to enable efficient remote memory access.
Co-design-based approaches co-design software and hardware~\cite{clio,concordia,MIND,rethinksoftruntime} to gain better application throughput and latency on DM.
The design of \kv is agnostic to the low-level implementations of all these DM approaches.

\noindent
\textbf{Disaggregated Memory Management.}
MIND~\cite{MIND} and Clio~\cite{clio} are the two state-of-the-art memory management approaches on DM.
But they both rely on special hardware to manage memory spaces.
The two-level memory management of \kv resembles the hierarchical memory management of The Machine~\cite{theMachine,hotos15themachine}.
The difference is that \kv focuses on fine-grained KV allocation with commodity RNICs, while The Machine relies on special SoCs and directly manages physical memory devices.

\noindent
\textbf{Memory-disaggregated KV stores.}
Clover~\cite{pdpm} and Dinomo~\cite{dinomo} are the most related memory-disaggregated KV stores.
Compared with Clover~\cite{pdpm}, \kv brings disaggregation to metadata management and gains better resource efficiency and scalability.
Finally, Dinomo~\cite{dinomo} is a fully-disaggregated KV store that was developed concurrently with our system. 
Dinomo proposes ownership partitioning to reduce coordination overheads of managing disaggregated metadata.
However, it assumes that the disaggregated memory pool is fault-tolerant, and hence its design does not consider MN failures.
In contrast, \kv addresses the challenges of handling MN failures with the \rep replication protocol.
There are many related RDMA-based KV stores~\cite{fasst,erpc,storm,pdpm,rfp,userdmaefficiently,ramcloud,xstore,drtm,drtmh,farm,pilaf}.
They are infeasible on DM since they rely on server-side CPUs to execute KV requests.
Besides, there are emerging approaches that use SmartNICs to construct KV stores~\cite{kvdirect,floem}.
\kv can also benefit from the additional compute power by offloading the memory management to SmartNICs.

\noindent
\textbf{Replication.}
Both traditional~\cite{cr,craq,primaryBackupRep,paxos,robustEmulation,epaxos,raft,quorum} and RDMA-based~\cite{hermes,mojim,tailwind} replication protocols are designed to ensure data durability.
However, all these approaches are server-centric replication protocols designed for monolithic servers.
Differently, \rep is a client-centric replication protocol designed for the DM architecture and achieves high scalability with collaborative conflict resolution.

%% file: Sections/08-conclusion.tex
\section{Conclusion}\label{sec:conclu}
\noindent
This paper proposes \kv, a fully memory-disaggregated KV store, that achieves both resource efficiency and high performance by disaggregating metadata management.
\kv adopts a client-centric replication protocol, a two-level memory management scheme, and an embedded log scheme to attack the challenges of weak MN-side compute power and complex failure situations on DM.
Experimental results show that \kv outperforms the state-of-the-art approaches by up to $4.5\times$ with less resource consumption.

%% file: Sections/09-acknowledgement.tex
\section*{Acknowledgments}
\noindent
We sincerely thank our shepherd Kimberly Keeton and the anonymous reviewers for their constructive comments and suggestions. 
This work was supported by the National Natural Science Foundation of China (Nos. 62202511 \& 61971145), the Research Grants Council of the Hong Kong Special Administrative Region, China (No. CUHK 14210920 of the General Research Fund), and Huawei Cloud.
Pengfei Zuo is the corresponding author (pfzuo.cs@gmail.com).

%% file: Sections/A-proof.tex
\section{Proof of Correctness}\label{sec:proof}
\noindent
Since \rep focuses on the linearizability of replicated index slots, it is equivalent to proving that \rep is correct on a single replicated slot.
The proof consists of four parts. 
First, we formally define the system model and the correctness standard, {\em i.e.}, linearizability.
We then specify the behavior of a legal replicated slot.
After that, we formally prove that \rep is correct under failure-free executions.
Finally, we illustrate that \rep is correct when failures occur with a fault-tolerant management master.

\subsection{Formal System Model}
\noindent
Formally, the system is comprised of a set $C$ of \textit{clients} $\{c_1, ..., c_n\}$ and a set $S$ of \textit{replicated slots} $\{s_1, ..., s_r\}$, where $n$ is the number of clients and $r$ is the replication factor.
Each slot $s_i$ stores a value denoted by $v_{s_i}$. 
Without loss of generality, we assume $s_1$ is the primary slot and $\{s_2, ..., s_r\}$ are backup slots.
\textit{Clients} can access the contents in the \textit{replicated slots} via $\texttt{RDMA\_READ}(s_i)$, and $\texttt{RDMA\_CAS}(s_i, v_{old}, v_{new})$ operations. 
$\texttt{RDMA\_READ}(s_i)$ returns the content in slot $s_i$. 
$\texttt{RDMA\_CAS}(s_i, v_{old}, v_{new})$ atomically compares $v_{s_i}$ with $v_{old}$ and updates $v_{s_i}' = v_{new}$ if $v_{s_i} = v_{old}$.
All these operations are synchronous and reliable, {\em i.e.}, client-initiated RDMA operations will not be re-ordered and operations reply a \texttt{FAIL} state only if the corresponding slot crashes.

Both \textit{clients} and \textit{slots} may fail according to the \textit{crash-stop model}: a failed client stops executing instructions and a failed slot returns \texttt{FAIL} to clients when performing RDMA operations.
Slots are shared memory that provide only $\texttt{RDMA\_READ}$ and $\texttt{RDMA\_CAS}$ operations for \textit{clients} to manipulate its content.
Clients are state machines that deterministically transition between states when events occur.
Each operation $op$ contains two events, $inv(op)$ indicating the invocation of the operation and $resp(op)$ indicating the completion of the operation.

We use the definitions of \textit{history} from Herlihy and Wing~\cite{DBLP:journals/toplas/HerlihyW90} to facilitate our proof.
A history $h$ of an execution $e$ is an infinite sequence of operation invocation and response events in the same order as they appear in $e$.
We denote by $ops(h)$ the set of all operations whose invocations appear in $h$.

A history $h$ is \textit{sequential} if (1) the first event of $h$ is an invocation and (2) each invocation, except the last is immediately followed by a matching response and each response is immediately followed by an invocation.
We use $h|o$ to denote an object subhistory, where all events in $h$ occurred at object $o$.
A set $S$ of histories is \textit{prefix-closed} if whenever $h$ is in $S$, every prefix of $h$ is also in $S$.
A \textit{sequential specification} of an object $o$ is a prefix-closed set of single-object sequential histories for $o$.
A sequential history $h$ is legal if $\forall o \in O: h|o$ belongs to the sequential specification for $o$.
A history induces an irreflexible partial order on $ops(h)$, denoted as $<_h$, where $op_1 <_h op_2$ if and only if $resp(op_1) < inv(op_2)$ in h.

\subsection{Replicated Slot}
\noindent
A replicated slot is a data type that supports the following operations:
\begin{itemize}
    \item READ(): returns the value of the replicated slot.
    \item WRITE(v): updates the value of the slot to v.
\end{itemize}

\noindent
Different from traditional replicated objects, the values different clients write to a replicated slot are not duplicated, which is guaranteed by the out-of-place modification scheme of \kv.
We use $reads(h)$ and $writes(h)$ to denote the set of all operations that reads or writes in $ops(h)$ respectively.

\begin{definition}[Replicated Slot Specification]\label{def:repslot-spec}
A sequential object history $h|o$ belongs to the sequential specification of a replicated slot if for each $op \in reads(h|o)$ such that $resp(op) \in h|o$, $resp(op)$ contains the value of the latest preceding operation $u \in write(h|o)$ or if there is no preceding update, then $resp(op)$ contains the initial value of $o$.
\end{definition}

\subsection{Proof of Linearizability}
\noindent
In this subsection, we prove that the failure-free executions of \rep satisfy linearizability.
The proof is based on Algorithm~\ref{alg:rep} and Algorithm~\ref{alg:rep-eval}.
Linearizability~\cite{lin} is defined as follows.

\begin{definition}[Linearizability]\label{def:lin}
A complete history $h$ satisfies linearizability if there exists a legal total order $\tau$ of $ops(h)$ such that $\forall op_1, op_2 \in ops(h), op_1 <_h op_2 \implies op_1 <_{\tau} op_2$.
\end{definition}

\textbf{Proof Sketch.}
We prove the linearizability of \rep by first attaching a \textit{virtual label} to the replicated slot.
Then we formally define the rules to update the slot virtual label.
We assign each READ or WRITE operation an operation label according to the slot virtual label.
Finally, we define a total order relationship on the operation virtual label and prove that the label total order satisfies $\tau$ in Definition~\ref{def:lin}.

\textbf{Notations}
Without further specification, we use $var^{op}$ and $func^{op}$ to denote the variable $var$ and the function call $func$ in the execution of the algorithms of $op$.

\textbf{More Definitions.}
We first give more definitions to facilitate our proof.
\begin{definition}
A complete operation $op \in reads(h)$ observes an update $u \in writes(h)$ if the value returned in $resp(op)$ was written by $u$.
\end{definition}

\begin{definition}[Write Winner]
An $op \in writes(h)$ is a write winner, denoted by $Win(op) = True$, if $win^{op} \in \{\texttt{Rule\_1}, \texttt{Rule\_2}, \texttt{Rule\_3}\}$ in Line~10 of Algorithm~\ref{alg:rep}.
\end{definition}

\begin{definition}[Slot Virtual Label]
A virtual label of a replicated slot is of the form $(u, v, w) \in \mathbb{N}^3$. The initial virtual slot label is $(0, 0, 0)$.
\end{definition}

\begin{definition}[Slot Virtual Label Update Rules]\label{def:label-upd}
For a label with value $(u, v, w)$ its update rules are defined as follows:
\begin{itemize}
    \item $(u, v, w) \leftarrow (u, v, w+1)$ for each \texttt{RDMA\_READ} operation in Line~2 of Algorithm~\ref{alg:rep}.
    \item $(u, v, w) \leftarrow (u+1, 0, 0)$ for each \texttt{RDMA\_CAS} that modifies the primary slot in Line~12 and Line~15 of Algorithm~\ref{alg:rep}.
\end{itemize}
\end{definition}

\begin{definition}[Operation Virtual Label]\label{def:ops-label}
For each operation $op \in ops(h)$, its virtual label $L_{op} = (u, v, w) \in \mathbb{N}^3$ is assigned as follows:
\begin{itemize}
    \item If $op \in reads(h)$, $L_{op} = (u, 0, w+1)$, where $(u, v, w)$ is the slot virtual label when it initiates the \texttt{RDMA\_READ} in Line~2 of Algorithm~\ref{alg:rep}.
    \item If $op \in writes(h), Win(op) = True$, then $L_{op} = (u+1, 0, 0)$, where $(u, v, w)$ is the slot virtual label when it initiates the \texttt{RDMA\_READ} in Line~6 of Algorithm~\ref{alg:rep}.
    \item If $op \in writes(h), Win(op) = False$, then $L_{op} = (u, 1, c_i)$, where $(u, v, w)$ is the slot virtual label when it initiates the \texttt{RDMA\_READ} in Line~6 of Algorithm~\ref{alg:rep} and $c_i$ is the client that performs the operation.
\end{itemize}
\end{definition}

\begin{definition}[Operation Label Order]
Two operation virtual labels $L_1 = (u_1, v_1, w_1), L_2 = (u_2, v_2, w_2)$, $L_1 <_l L_2$ if and only if $u_1 < u_2 \bigvee (u_1 = u_2 \bigwedge v_1 < v_2) \bigvee (u_1 = u_2 \bigwedge v_1 = v_2 \bigwedge w_1 < w_2)$.
\end{definition}

\begin{definition}[Write Round]
For an $op \in writes(h)$, its write round $R_{op} = u$, where $u$ is the label $(u, v, w)$ when Line~6 of Algorithm~\ref{alg:rep} is executed.
\end{definition}

\begin{definition}[Round Start Time]\label{def:r-st}
The start time of a write round $k$ is defined as $t_{start}^k = inv(op): \forall op_j \in writes(h): R(op) = R(op_j) \implies inv(op) \leq inv(op_j)$.
\end{definition}

\begin{lemma}\label{lemma:1}
Values in the primary slot is uniquely associated with a write round.
\end{lemma}
\begin{proof}
By Definition~\ref{def:label-upd} and Line~12 and Line~15 of Algorithm~\ref{alg:rep}. The primary slot is modified with different values and with a strictly monotonic $u$, where $(u, v, w)$ is its virtual label.
\end{proof}

\begin{lemma}[At least one winner]\label{lemma:>=1winner}
$\forall$ write round $k$ if $v_{s_1} = ... = v_{s_r}$ at $t_{start}^k$, then $\exists op \in writes(h): R(op) = k \bigwedge win(op) = True$.
\end{lemma}
\begin{proof}
Since $v_{s_1} = ... = v_{s_r}$, by Lemma~\ref{lemma:1} $\forall op \in writes(h): R(op) = k, v_{old}^{op} = v_{s_1}$.
Since values in the backup slots all equals $v_{s_1}$ and different $op$ proposes different $v_{new}^{op}$, the atomicity of \texttt{RDMA\_CAS} (Line~7, Algorithm~\ref{alg:rep}) ensures that each backup can only be modified once and all backups must have be modified once after some $op$ reaches Line~9.
Suppose $\forall op \in writes(h): R(op) = k \bigwedge Win(op) = False$.
The primary slot will not be modified as indicated by the if condition in Line~16 of Algorithm~\ref{alg:rep}, which implies that all $op$ has $v_{check}^{op} = v_{old}$ in Line~12 of Algorithm~\ref{alg:rep-eval}.
The if condition of Line~17 (Algorithm~\ref{alg:rep-eval}) must be false, and thus $\forall op: v_{new}^{op} \neq min(v\_list)$, which is impossible because all the backup slot must have be modified by some $op$ and $v_{new}^{op}$ are distinct.
\end{proof}

\begin{lemma}[At most one winner]\label{lemma:<=1winner}
$\forall$ write round $k$ if $v_{s_1} = ... = v_{s_r}$ at $t_{start}^k$, then $\forall op_1, op_2 \in writes(h): R(op_1) = R(op_2) = k \bigwedge Win(op_1) = Win(op_2) = True \implies op_1 = op_2$.
\end{lemma}
\begin{proof}
1. No two client can both win by \texttt{Rule 1} and \texttt{Rule 2} because (1) all backup slots can only be modified once and (2) majority cannot be overlapped because different $op$ modifies with different $v_{new}$.

2. Suppose there are two winners $op_1$ and $op_2$.

\noindent
\textbf{CASE 1}: $win^{op_1} \in \{\texttt{Rule\_1}, \texttt{Rule\_2}\} \bigwedge win^{op_2} = \texttt{Rule\_3}$.

\noindent
If the finish of \texttt{RDMA\_CAS\_backups} (Line~7, Algorithm~\ref{alg:rep}) of $op_2$ happens before $op_1$ executes \texttt{RDMA\_CAS\_primary} (Line~12 or Line~15, Algorithm~\ref{alg:rep}), then guaranteed by the atomicity of \texttt{RDMA\_CAS}, $op_2$ knows the majority and cannot win.
If the finish of \texttt{RDMA\_CAS\_backups} of $op_2$ happens after $op_1$ executes \texttt{RDMA\_CAS\_primary}, then the \texttt{RDMA\_READ} (Line~12, Algorithm~\ref{alg:rep-eval}) must return $v_{check}^{op_2} \neq v_{orig}^{op_2}$. Then $win^{op_2} = \texttt{FINISH}$ by Line~16 of Algorithm~\ref{alg:rep-eval}.

\noindent
\textbf{CASE 2}: $win^{op_1} = win^{op_2} = \texttt{Rule 3}$.

\noindent
Since there is no majority, Line~17 of Algorithm~\ref{alg:rep-eval} indicates that $v_{new}^{op_1} = v_{new}^{op_2}$, contradicting with the assumption that no clients write duplicated value.
\end{proof}

\begin{definition}[Round finish time]
The round finish time $t_{fini}^k$ of a write round $k$ is the time when its only winner initiates \texttt{RDMA\_CAS\_primary} in Line~12 or Line~15 of Algorithm~\ref{alg:rep}.
\end{definition}

\begin{lemma}\label{lemma:consistent-round-init}
$\forall k$ at $t_{start}^k$, we have $v_{s_1} = ... = v_{s_r}$.
\end{lemma}
\begin{proof}
We prove this by induction.

1. Initially, when $k = 0$, all slots have the same value.

2. Suppose all slots have the same value at $t_{start}^k$.
By Lemma~\ref{lemma:>=1winner} and Lemma~\ref{lemma:<=1winner}, there must be a single write winner $op$.
Guaranteed by the if-condition in Line~11 and Line~13 of Algorithm~\ref{alg:rep}, only the $op$ with $Win(op) = True$ can further modify the replicated slot.
If $win^{op} = \texttt{Rule\_1}$, then Line~7 of Algorithm~\ref{alg:rep} guarantees that all backups $v_{s_2} = ... = v_{s_r} = v_{new}^{op}$ before the initiation of \texttt{RDMA\_CAS\_primary} (Line~12).
If $win^{op} \in \{\texttt{Rule\_2}, \texttt{Rule\_3}\}$, then the \texttt{RDMA\_CAS\_backups} (Line~14, Algorithm~\ref{alg:rep}) guarantees that all backups $v_{s_2} = ... = v_{s_r} = v_{new}^{op}$ before the \texttt{RDMA\_CAS\_primary} (Line~15).
Since $op$ modifies the primary slot also to $v_{new}^{op}$ (Line~12 or Line~15), after $t_{fini}^0$ we have $v_{s_1} = ... = v_{s_r} = v_{new}^{op}$.
By Definition~\ref{def:r-st}, all modifications happens after $t_{start}^{k+1}$, which implies $v_{s_1} = ... = v_{s_r}$ at $t_{start}^{k+1}$. 
\end{proof}

\begin{lemma}[Exactly one write winner]\label{lemma:=1winner}
$\forall k \in \mathbb{N}$, we have:
\begin{itemize}
    \item $\exists op \in writes(h): R(op) = k, Win(op) = True$.
    \item $\forall op_1, op_2 \in writes(h): R(op_1) = R(op_2) = k \bigwedge Win(op_1) = Win(op_2) = True \implies op_1 = op_2$.
\end{itemize}
\end{lemma}
\begin{proof}
By Lemma~\ref{lemma:>=1winner} + Lemma~\ref{lemma:<=1winner} + Lemma~\ref{lemma:consistent-round-init}.
\end{proof}

\begin{lemma}[Label order respects history]\label{lemma:<l-respects-<h}
$\forall op_1, op_2 \in ops(h), op_1 <_h op_2 \implies op_1 <_l op_2$.
\end{lemma}
\begin{proof}
Suppose $\exists op_1, op_2 \in ops(h), op_1 <_h op_2 \bigwedge op_2 <_l op_1$. Their labels are $l_1 = (u_1, v_1, w_1)$ and $l_2 = (u_2, v_2, w_2)$.

\noindent
\textbf{CASE 1}: $op_1, op_2 \in reads(h)$.

\noindent
If $u_1 = u_2$, then $op_1 <_h op_2$ implies that the \texttt{RDMA\_READ} (Line~2, Algorithm~\ref{alg:rep}) of $op_1$ happens before the \texttt{RDMA\_READ} of $op_2$. 
By Definition~\ref{def:label-upd}, \texttt{RDMA\_READ}s in Line~2 of Algorithm~\ref{alg:rep} monotonically increase $w$ of the virtual slot label. Then we must have $w_1 < w_2$. Since $u_1 = u_2, v_1 = v_2 = 0, w_1 < w_2$, $op_1 <_l op_2$, contradicts with the assumption.
If $u_1 > u_2$, then $\exists t_{fini}^{u_2}: resp(op_1) < t_{fini}^{u_2} < inv(op_2)$. 
By Definition~\ref{def:label-upd}, $u_1 \leq u_2$, contradicts with $u_1 > u_2$.

\noindent
\textbf{CASE 2}: $op_1 \in reads(h), op_2 \in writes(h)$.

\noindent
If $Win(op_2) = True$, then $op_1 <_h op_2$ implies that $resp(op_1) < t_{fini}^{u_2 - 1}$. Then $u_1 \leq u_2 - 1$ contradicts with $op_2 <_l op_1$.
Otherwise $Win(op_2) = False$, then $op_1 <_h op_2$ implies that $resp(op_1) < t_{fini}^{u_2}$. Since $u_1 \leq u_2$ and $0 = v_1 < v_2 = 1$ by Definition~\ref{def:ops-label}, there is a contradiction with $op_2 <_l op_1$.

\noindent
\textbf{CASE 3}: $op_1 \in writes(h), op_2 \in reads(h)$.

\noindent
Line~17-22 of Algorithm~\ref{alg:rep} ensures that $\forall op \in writes(h), R(op) = k \implies resp(op) > t_{fini}^k$.
If $Win(op_1) = True$, then $op_1 <_h op_2$ implies that $t_{fini}^{u_1-1} < inv(op_2)$ and $l_1 = (u_1, 0, 0)$. Consequently, $u_1 \leq u_2 \bigwedge v_1 = v_2 = 0 \bigwedge 0 = w_1 < 1 \leq w_2$ contradicts with $op_2 <_l op_1$.
Otherwise $Win(op_1) = False$, then $op_1 <_h op_2$ implies that $t_{fini}^{u_1} < inv(op_2)$. Then $u_1 + 1 \leq u_2$ which contradicts with $op_2 <_l op_1$.

\noindent
\textbf{CASE 4}: $op_1, op_2 \in writes(h)$.

\noindent
If $Win(op_1) = True$, then $op_1 <_h op_2 \implies t_{fini}^{u_1 - 1} < inv(op_2) \implies u_1 \leq u_2$. Also $Win(op_1) = True \implies l_1 = (u_1, 0, 0)$. Then $u_1 \leq u_2, 0 = v_1 < v_2 = 1$ contradicts with $op_2 <_l op_1$.
Otherwise $Win(op_1) = False$, then $op_1 <_h op_2 \implies t_{fini}^{u_1} < inv(op_2) \implies u_1 + 1 \leq u_2 \implies u_1 < u_2$, which contradicts with $op_2 <_l op_1$.
\end{proof}

\begin{lemma}[Label order is a legal order]\label{lemma:legal<l}
The label order $<_l$ is a legal total order of $ops(h)$.
\end{lemma}
\begin{proof}
By Definition~\ref{def:repslot-spec}, let $r \in reads(h)$ be an operation that observes a write $k \in writes(h)$, $r$ is completed.
Suppose $\exists k' \in writes(h)$ such that $k <_l k' <_l r$. 
Let $l_k = (u_k, v_k, w_k), l_{k'} = (u_k', v_k', w_k'), l_r = (u_r, v_r, w_r)$.
By Lemma~\ref{lemma:1}, $r$ observes a value if and only if it returns the corresponding virtual label in Line 2 of Algorithm~\ref{alg:rep}.
By Definition~\ref{def:label-upd} and Lemma~\ref{lemma:=1winner}, the label is exclusively updated by a single writer, which implies $Win(k) = True, t_{fini}^{u_k - 1} < inv(r)$ and $u_k = u_r$.
Since $k <_l k' <_l r$ and $u_k = u_r$, $u_k = u_k'= u_r$.
If $Win(k') = True$, then $Win(k) = Win(k') = True \bigwedge R(k) = R(k') = u_k - 1$, which contradicts with Lemma~\ref{lemma:=1winner}.
Otherwise $Win(k') = False$ then by Definition~\ref{def:ops-label} $v_k' = 1$.
$u_k' = u_r, v_k' = 1 > 0 = v_r$ contradicts with $k' <_l r$.
As a result, there is no $k'$ such that $k <_l k' <_l r$.
\end{proof}

\begin{theorem}
\kv implements a replicated slot with linearizability.
\end{theorem}
\begin{proof}
By Lemma~\ref{lemma:legal<l} and Lemma~\ref{lemma:<l-respects-<h}.
\end{proof}

\input{Algorithms/SNAPSHOT-full-1}
\input{Algorithms/SNAPSHOT-full-2.tex}

\subsection{Correctness of Failure Recovery}
\noindent
The recovery of \kv relies on a fault-tolerant management master, which is a common assumption on membership-based replication protocols like Chain Replication~\cite{cr,craq} and Hermes~\cite{hermes}.
The master adopts a lease-based membership service~\cite{ukharon} for clients and MNs.
The lease-based membership service provides a view of all alive MNs to clients.
Clients check and extend their leases before performing each \textit{read} and \textit{write}.
The master can detect the failures of clients and MNs when a client or MN no longer extend their leases.
The pseudo-code of the master is shown in Algorithm~\ref{alg:mastr}, and the full process of clients is shown in Algorithm~\ref{alg:rep-full}.
The key to guaranteeing correctness under failures is that the single-winner condition is not violated.

\subsubsection{Handling MN Crashes}
\noindent
The failure handling process consists of three phases, {\em i.e.}, a disconnection phase (Line 2-4, Algorithm~\ref{alg:mastr}), a slot-modification phase (Line 5-7, Algorithm~\ref{alg:mastr}), and a notification phase (Line 8-10, Algorithm~\ref{alg:mastr}).
The disconnection phase starts with the master sending a \textit{member\_prepare\_change} to all clients and waiting for a lease expiration.
On receiving the \textit{member\_prepare\_change}, clients stop their future \textit{write}s to the failed slot.
After the lease expires, no clients can further modify the failed slot because the failed slot is excluded from the membership view.
The currently executing write operations will also be stopped and wait for the final value to be decided by the master (Line~35-28, Algorithm~\ref{alg:rep-full}).
Hence, all clients are disconnected to the crashed slot.

During the modification phase, the master randomly selects a value in an alive backup slot and uses the selected value to make all alive slots consistent.
Choosing values from backup slots is always safe because backup slots contain no older values than the latest committed value (value in the primary slot before it crashed).
This is guaranteed by the \rep replication protocol that the write conflicts are resolved in the backup slots and then written to the primary slot.
If there are no alive backup slots, then there must have only one survivor as the primary slot.
In both cases, either a latest committed value or a fresh uncommitted value (value written by conflicting writers) is chosen.
When a fresh uncommitted value is chosen, the master is the representative last writer and finishes the last writer's job on a client's behalf.
Since no other client can modify the slot, the master is the only last writer.
When a committed value is chosen, the value will be replied to clients and clients will retry their newer WRITEs (Line~38, Algorithm~\ref{alg:rep-full}) on receiving an old value.
In this case, the new write round of clients is considered not started.
The retry of \textit{write}s guarantees that the clients' newer write operations are executed before they are returned.
The single last write will then be decided among clients through the normal execution of the \rep replication protocol.
After modifying all the values in the replicated slots, the master commits the operation by writing a special value ($1$) to the \textit{old value} field of the log header.
This guarantees that clients will never retry an operation finished by the master.

The attribute of exactly one winner on each write round is not violated.
Line~4 of Algorithm~\ref{alg:rep-eval} ensures that a winner can only be decided when no backup slot crashes.
Line~29 and Line~35-38 of Algorithm~\ref{alg:rep-full} ensure that all losers and failed operations wait for the decided value returned by the master.
If there is a winner and the winner finishes execution before the disconnection phase of the master, then the master can only choose no older value than the winner's committed value because all old values are modified to the winner's proposed new value.
If the winner does not finish execution before the disconnection phase of the master, then it will be disconnected and deemed also as a loser.
In this case, the master becomes a representative winner that decides a unique winner value to all other waiting losers.
If there is no winner due to the backup slot failures, then the master will also become a representative winner that decides a unique winner.

During MN crashes, \texttt{read}s can also get the latest committed value in the slot, thus ensuring linearizability.
The latest committed value is defined as the last value stored in the primary slot before it crashes.
Since the primary slot is lastly modified in a write round (Line~21 and Line~24, Algorithm~\ref{alg:rep-full}), it always contains the latest committed value. 
If the primary slot is alive, then \texttt{read}s execute normally by reading the content in the primary slot (Line 2, Algorithm~\ref{alg:rep-full}).
Linearizability is guaranteed by the \rep replication protocol.
Otherwise, \texttt{read}s use \texttt{RDMA\_READ}s to read all backup slots and return a value only if all the values in the backup slot are the same.
This reflects the situation when there are no write conflicts, and hence the value must be committed.
Under the situation when values in backup slots are not the same, clients send RPCs to the master to let the master commit a value and return it to clients.

\subsubsection{Handling Client Crashes}
\noindent
\kv uses per-client operation logs to recover crashed client operations.
The $v_{old}$ is written to the log header by the write winner before the execution of \texttt{RDMA\_CAS\_primary} (Line~21 and Line~24, Algorithm~\ref{alg:rep-full}).
$0$ is written to the \textit{used bit} of the log headers before the losers return their operations.
During the recovery of a crashed client, memory objects with unset \textit{used} bits are reclaimed because they are either incomplete data or free data.
Operations in the operation log with an incomplete $v_{old}$ are redone.
These operations may belong to a loser or a winner.
Redoing a crashed operation of a loser is safe because the losers' operations have not been returned to clients and cannot be observed by other clients.
Redoing such operations of a winner is also safe because the winner's operation also cannot be observed by other clients due to the unmodified primary slot.
Under this situation, the retried operation will also become a winner due to Lemma~\ref{lemma:=1winner}.

If the \textit{old value} field is set, then the crashed client must be a write winner since only the winner commits the log~\ref{sec:embedlog}.
However, for a winner, it is possible that it crashed before the primary slot was modified.
Under this situation, all backup slots contain $v_{new}$, and the primary slot contains the $v_{old}$.
Hence, \kv checks if the primary slot equals $v_{old}$ and modifies the primary slot to be $v_{new}$ if it is the case.
Otherwise, the operation must have been finished because the primary slot is modified in a new write round.

\subsubsection{Handling Client and MN Crashes}
\noindent
When handling concurrent client and MN failures, the master reconfigures MNs to use a decreased replication factor.
After the reconfiguration, the master recovers the crashed client.
The only conflict between recovering MNs and clients is the case that the master decides the value written by a client, and the client is crashed.
In this situation, a winner client is crashed and ignorant of its winning.
The operations may be retired, causing old values to be written to the slots.
\kv addresses this issue by letting the master commit the operation by writing a $v_{old} = 0$ to the \textit{old value} field of the log header.
As a result, when recovering the winner client using the log, it will not execute the operation twice.

%% file: Algorithms/SNAPSHOT-full-1.tex
\begin{algorithm}[H]
\begin{algorithmic}[1]
\caption{The master process.}\label{alg:mastr}
\Procedure{Master}{}
  \If {MN failed or received client fail\_query}
    \State \textit{send} \texttt{member\_prepare\_change} to all clients
    \State \textit{wait} all clients reply or membership lease expires
    \State \textit{select} a slot $s$ randomly in alive backups
    \State \textit{modify} all slots $s_i = s$
    \State \textit{commit} the operation log of $write(s)$
    \State \textit{select} new primary and backup
    \State \textit{reply} all clients' fail\_query with a new value $s$
    \State \textit{send} \texttt{member\_commit\_change} to all clients with the new membership
  \EndIf
  \If {Clients failed}
    \State \textit{start} \texttt{RecoverClient}($log$) thread.
  \EndIf
\EndProcedure
\Procedure{RecoverClient}{$log$}
  \If {The log is committed}
    \State \Return{}
  \EndIf
  \State $slot$ is the slot in the log
  \State $v_{old}$ is the \textit{old value} in the log
  \State $v_{new}$ is the new value in the log
  \State $s_0 = \texttt{RDMA\_READ\_primary}(slot)$
  \If {$\texttt{FAIL} \in \{s_0, bk\_list\}$}
    \State \textit{send} FailReq to master.
  \ElsIf {$v_{old}$ is incomplete}
    \State \textit{retry} the operation
  \ElsIf {$v_{old} = s_0$}
    \State $\texttt{RDMA\_CAS\_primary}(slot, s_0, v_{new})$
  \EndIf
  \State \Return{}
\EndProcedure
\end{algorithmic}
\end{algorithm}

%% file: Algorithms/SNAPSHOT-full-2.tex
%
\begin{algorithm}[h]
\begin{algorithmic}[1]
\caption{\rep with failure handling}\label{alg:rep-full}
\small
\Procedure{READ}{$slot$}
  \State $v=\texttt{RDMA\_READ\_primary}(slot)$
  \If {$v=\texttt{FAIL}$}
    \State $v\_list=\texttt{RDMA\_READ\_backups}(slot)$
    \If {All backups have the same value $v'$}
      \State $v=v'$
    \Else
      \State $v = \texttt{RPC\_fail\_query}(master, slot)$
      \State \textit{wait} for membership change
    \EndIf
  \EndIf
  \State \Return{$v$}
\EndProcedure
\Procedure{WRITE}{$slot,v_{new}$}
  \State $v_{old} = \texttt{RDMA\_READ\_primary}(slot)$
  \If {$v_{old} = \texttt{FAIL}$}
    \State \textit{wait} for membership change
    \State \textit{retry} WRITE
  \EndIf
  \State $v\_list = \texttt{RDMA\_CAS\_backups}(slot, v_{old}, v_{new})$
  \LineComment{Change all the $v_{old}$s in the $v\_list$ to $v_{new}$s.}
  \State $v\_list = \texttt{change\_list\_value}(v\_list, v_{old}, v_{new})$
  \State $win=\texttt{evaluate\_rules}(v\_list)$
  \If {$win=\texttt{Rule\_1}$}
    \State $\texttt{RDMA\_CAS\_primary}(slot, v_{old}, v_{new}$
  \ElsIf {$win \in \{\texttt{Rule\_2}, \texttt{Rule\_3}\}$}
    \State $\texttt{RDMA\_CAS\_backups}(slot, v\_list, v_{new})$
    \State $\texttt{RDMA\_CAS\_primary}(slot, v_{old}, v_{new})$
  \ElsIf {$win = \texttt{LOSE}$}
    \Repeat
      \State sleep a little bit
      \State $v_{check}=\texttt{RDMA\_READ\_primary}(slot)$
      \If {\textit{receive} \texttt{member\_prepare\_change}()}
        \State goto Line~35
      \EndIf
    \Until {$v_{check} \neq v_{old}$}
    \If {$v_{check} = \texttt{FAIL}$}
      \State goto Line~35
    \EndIf
  \ElsIf {$win = \texttt{FAIL}$}
    \State $v_{RPC} = \texttt{RPC\_fail\_query}(master, slot)$
    \State \textit{wait} for membership change
    \If {$v_{RPC} = v_{old}$}
      \State \textit{retry} WRITE
    \EndIf
  \EndIf
  \State \Return{}
\EndProcedure
\end{algorithmic}
\end{algorithm}